\documentclass[11pt]{article}
\usepackage[letterpaper,margin=1in]{geometry}
\usepackage[T1]{fontenc}
\usepackage[utf8]{inputenc}
\usepackage{amsmath,amssymb,amsthm,mathrsfs,bm,bbm,dsfont}
\usepackage{graphicx}
\usepackage{xcolor}
\usepackage{booktabs,multirow}
\usepackage{subcaption}
\usepackage{enumitem}
\usepackage{siunitx}
\usepackage{placeins}
\usepackage{algorithm}
\usepackage{algpseudocode}
\usepackage[round,authoryear]{natbib}
\usepackage{url}
\usepackage[colorlinks=true,linkcolor=blue,citecolor=blue,urlcolor=blue]{hyperref}

\newcommand{\beq}{\begin{equation}}
\newcommand{\eeq}{\end{equation}}
\newcommand{\beqa}{\begin{eqnarray}}
\newcommand{\eeqa}{\end{eqnarray}}
\newcommand{\beqan}{\begin{eqnarray*}}
\newcommand{\eeqan}{\end{eqnarray*}}

\newcommand{\vnorm}[1]{\left\|#1\right\|}

\newcommand{\E}{\mathds{E} }

\newcommand{\Nset}{\mathbb{N}}

\newcommand{\Rset}{\mathbb{R}}

\newcommand{\Tset}{\mathbb{T}}

\newcommand{\Xset}{\mathbb{X}}

\newcommand{\Bcal}{{\cal B}}

\newcommand{\Fcal}{{\cal F}}

\newcommand{\Hcal}{{\cal H}}
\newcommand{\Ical}{{\cal I}}
\newcommand{\Jcal}{{\cal J}}
\newcommand{\Kcal}{{\cal K}}

\newcommand{\Pcal}{{\cal P}}

\newcommand{\Ucal}{{\cal U}}

\renewcommand{\[}{\left[}

\newcounter{l1}
\newcounter{l2}
\newcounter{l3}
\newcommand{\bdotlist}{\begin{list}{$\bullet$}{}}
\newcommand{\bboxlist}{\begin{list}{$\Box$}{}}
\newcommand{\bbboxlist}{\begin{list}{\raisebox{.005in}{{\tiny
$\blacksquare$ \ \ }}}{}}
\newcommand{\bdashlist}{\begin{list}{$-$}{} }
\newcommand{\blist}{\begin{list}{}{} }
\newcommand{\barablist}{\begin{list}{\arabic{l1}}{\usecounter{l1}}}
\newcommand{\balphlist}{\begin{list}{(\alph{l2})}{\usecounter{l2}}}
\newcommand{\bAlphlist}{\begin{list}{\Alph{l2}.}{\usecounter{l2}}}
\newcommand{\bdiamlist}{\begin{list}{$\diamond$}{}}
\newcommand{\bromalist}{\begin{list}{(\roman{l3})}{\usecounter{l3}}}

\newtheorem{theorem}{Theorem}
\newtheorem{lemma}{Lemma}
\newtheorem{proposition}{Proposition}
\newtheorem{example}{Example}
\newtheorem{corollary}{Corollary}

\newtheorem{assumption}{Assumption}

\newcommand{\R}{\mathbb{R}}
\newcommand{\dd}{\mathrm{d}}

\newcommand{\DKL}{D_{\mathrm{KL}}}
\newcommand{\tE}{\tilde{\mathbb{E}}}
\newcommand{\tP}{\tilde{P}}
\newcommand{\grad}{\nabla}

\newcommand{\Lmd}{L}

\newcommand{\Pmd}{P}

\title{Transition Path Sampling Using Koopman Operators and Exit-Time Optimal Control}

\author{Boya Hou\thanks{Co-first authors.} \and Shane Wang\footnotemark[1] \and Siddharth Ambekar \and Maxim Raginsky \and Olgica Milenkovic\\[6pt]  University of Illinois Urbana-Champaign\\}

\begin{document}

\maketitle
\begin{abstract}
Sampling transitions between metastable states is a central problem in dynamical systems theory and molecular dynamics in particular. A key challenge is the existence of high free-energy barriers that separate the states, making transitions extremely rare. Recent machine learning-based methods cast transition path sampling (TPS) as an optimal stochastic control (OSC) problem over a fixed time horizon, and parameterize the drift bias via a neural network trained by simulation-in-the-loop, requiring repeated biased rollouts. To address computational and performance guarantee issues of these models, we propose a new approach for the problem based on Koopman operators. Because Koopman operators are linear, their leading eigenfunctions reveal the metastable sets and provide an estimate of the committor function with no transition path information required. Furthermore, we formulate TPS as an OSC problem up to an \emph{exit time}. Our time horizon is the first hitting time of the target set, and our running cost penalizes time spent in nonreactive regions by encoding the estimated committor function. We derive the optimal controller in closed form and approximate it in a reproducing kernel Hilbert space (RKHS). This reduces the problem of constructing the optimal controller to solving a single equality-constrained quadratic program, whose solution can be characterized by a linear Karush-Kuhn-Tucker (KKT) system. On the two-channel double well and alanine dipeptide, our controller increases the fraction of trajectories reaching the target from $0\%$ to $99.8\%$ within $1{,}000$ steps, and from $0\%$ to $93\%$ within $1\mathrm{ps}$, respectively. 

\end{abstract}

\section{Introduction}
\label{sec:intro}
Sampling transitions between metastable sets is a fundamental problem in dynamical systems theory and molecular dynamics, underlying the study of conformational changes of proteins, folding, and drug discovery~\citep{PhysRevE.52.2893,PhysRevLett.78.3908}. A typical challenge is that the states of interest are separated by high free-energy barriers: the overwhelming majority of simulated trajectories remain within a single metastable set, while trajectories that cross from one metastable set to another without returning (reactive paths) are rare under native dynamics~\citep{vanden2006towards}.

Classical transition path sampling (TPS) performs Markov chain Monte Carlo in the path space, generating new trajectories from existing ones \citep{dellago1998efficient,bolhuis2002transition,bolhuis2021tps}. However, they require an initial reactive path connecting the two metastable sets, which can be difficult to generate. 
In physical chemistry and molecular dynamics, rare event sampling is typically accelerated by adding bias potentials to the energy landscape~\citep{PhysRevLett.78.3908,PhysRevE.52.2893,doi:10.1073/pnas.202427399}. These are usually defined over a small number of hand-crafted collective variables. Modern machine learning-based samplers such as PIPS~\citep{holdijk2023pips} and TPS-DPS~\citep{seong2024tpsdps} cast transition path sampling as optimal stochastic control or path-measure matching problems over a fixed time horizon, and parameterize the drift bias with a neural network trained by simulation-in-the-loop, which calls for repeated biased rollouts. Learning a bias force through simulation, however, can be challenging. For example, early in the training, few rollouts reach the target, and hence the training signal is highly sparse. Furthermore, a fixed-horizon formulation requires a properly chosen horizon, and the transition time is a random variable dependent on the path itself. Partial mitigation of these drawbacks is possible using replay buffers and annealing, and truncating the paths at a time that maximizes specific kernel values.

To resolve the above issues, we propose an \emph{operator-based} framework for efficient TPS that uses \emph{Koopman spectral analysis} and optimal stochastic control (OSC) up to an \emph {exit time}, in which the stopping rule is the first hitting time of the target set. Specifically, we cast TPS as sampling from a Gibbs measure on the path space problem, following the information-theoretic frameworks of~\cite{mitter2003variational} and
~\cite{raginsky2026}. By the Gibbs variational principle, the Gibbs measure associated with a reference path measure and an energy on a path is, under certain conditions, the unique minimizer of the free energy. \cite{mitter2003variational} used this insight to pose nonlinear filtering as an optimal control problem. To adapt the described sampling method to TPS, we take the reference to be the unbiased path measure of the original dynamics, and the energy to penalize the time spent in nonreactive regions via the committor function. Since we assume the time horizon to be the first time a path reaches the target set, we operate with exit rather than fixed times. We prove that the resulting Gibbs measure is the law of a controlled diffusion process and prove that the optimal controller minimizes the free energy over \emph{all} path measures on the stopped $\sigma$-algebra rather than only those generated by admissible controls. Furthermore, the optimal controller admits a closed-form solution and can be computed without simulations in the loop.

It is worth pointing out that the proposed approach differs from the recent work of~\cite{du2026rare}, which investigates a related hitting-time problem for committor estimation, treating the committor as a value function and training it with neural networks that require repeated simulation of controlled dynamics during training. As previously emphasized, we instead follow an \emph{operator-theoretic} approach in which the spectrum of the Koopman generator of the model reveals the metastable sets, and where the committor function is constructed from Koopman eigenfunctions. These can be estimated from a fixed dataset without requiring any trajectory to transition between the metastable sets. The Koopman generator plays a central role within this framework: its leading eigenfunctions reveal the metastable sets and are incorporated into the committor function, which in turn defines the running cost of the OSC. The optimal controller of the OSC can be solved via a boundary value problem wrt the operator. This pipeline comes with several practical advantages. First, it requires no transition paths, as Koopman eigenfunctions can be computed without crossings between metastable sets via a generalized eigenproblem. Second, the controller admits a closed-form solution, and thus is obtained without network training but through solutions of convex quadratic programs. The approach is also interpretable, as the controller is built from physically meaningful quantities that characterize the mechanism of the transition, namely the metastable sets and the committor function. Furthermore, the sampled paths can be reweighted to recover the statistics of the original dynamics.

Our contributions are multifold: (1) We formulate TPS as an OSC up to an exit time. We show in Theorem~\ref{thm:exit-time-SOC} that the optimal controller admits a closed-form solution, $u^* = \nabla \log \rho$, and that its path measure minimizes the free energy over \emph{all} path measures, not only those generated by admissible controls; 
(2) We construct this controller by approximating $\rho$ in a Reproducing Kernel Hilbert Space (RKHS). The problem then reduces to solving a convex quadratic program with linear equality constraints, whose solution can be characterized by a single linear Karush–Kuhn–Tucker (KKT) system rather than repeated rollouts; 
(3) In TPS, the source and target metastable sets are often predefined through collective variables. Following Koopman-based approaches that identify metastable sets~\citep{mardt2018vampnets,perez2013identification}, we construct these sets without collective variables by using the leading eigenfunctions of the Koopman generator, which capture the slow inter-basin modes. From the same eigenfunctions, we construct a committor function, which defines the running cost of the {OSC}; (4) We generate transition paths by simulating the controlled diffusion process until it first enters the target sets, and show the resulting paths can be reweighted via Lemma~\ref{lemma:girsanov-density} to recover the statistics of the original dynamics up to the transition; (5) We evaluate our method on both synthetic and real-world datasets.

The paper is organized as follows. In Section~\ref{sec:preliminaries}, we introduce the diffusion process of interest and the corresponding Koopman generator. In Section~\ref{sec:exit-time}, we develop the {OSC} up to an exit time in general form. There, we also derive a closed form of the optimal controller. The results of this section can be specialized for efficient TPS. Section~\ref{sec:opt-controller} constructs metastable sets and the committor function from Koopman eigenfunctions, and approximates the optimal controller in an RKHS. We generate transition paths by simulating the optimally controlled SDE, and report numerical experiments in Section~\ref{sec5}. Due to space limitations, proofs and extra experiments are deferred to the Appendix~\ref{sec:experiments}.

\section{Preliminaries: It\^o Diffusion and the Koopman Generator}
\label{sec:preliminaries}
Let $\Xset \subseteq \Rset^n$ denote the state space and consider the It\^o diffusion process
\begin{align}
  dX_t = b(X_t)\,dt + \sigma(X_t)\,dW_t, \qquad X_0 = x,
  \label{eq:ref-sde}
\end{align}
where $W_t$ is an $m$-dimensional Brownian motion process, and $b:\Xset\to\R^n$ and $\sigma:\Xset\to\R^{n\times m}$ are the drift and diffusion functions, respectively~\cite{karatzas2014brownian,oksendal2013stochastic}. A model of the form~\eqref{eq:ref-sde} used in MD simulations is the \emph{Langevin dynamics,} where $b$ and $\sigma$ are specified by a potential energy $U$, friction $\gamma$, and bath temperature $T_{\mathrm{emp}}$. In the overdamped case, the state is the atomic configuration, and the drift and diffusion coefficients are specialized as $b(x) = -\tfrac{1}{\gamma m}\nabla U(x)$, $\sigma(x) = \sigma_0 I$, with $\sigma_0 = \sqrt{\tfrac{2 k_B T_{\mathrm{emp}}}{\gamma m}}$. In the underdamped case, the state additionally carries the momenta, and friction and noise act only on the momentum component. In either case, the process is metastable whenever $U$ has multiple deep minima (wells) separated by barriers large relative to $k_B T_{\mathrm{emp}}$, and a trajectory spends most of its time in a single well.

Next, let $\Tset:=[0,\infty)$ and fix the space $\Omega = C(\Tset;\Xset)$ of continuous paths mapping $\Tset$ to $ \Xset$. Also, let $X_t(\omega)=\omega(t)$ denote the evaluation of a path $\omega\in\Omega$ at time $t \in \Tset$. Let $\Fcal=\sigma(X_t : t\in\Tset)$ be the canonical $\sigma$-algebra, and let $(\Fcal_t)_{t\in \Tset}$ with $\Fcal_t=\sigma(X_s : s\le t)$ be the filtration, both generated by the process $(X_t)_{t\in\Tset}$. 
For $x \in \Xset$, let $P^x$ be the law on $(\Omega,\Fcal)$ of the solution of \eqref{eq:ref-sde} started at $x$, provided it is well posed, for example, under Assumptions~\ref{A1} and \ref{A1.elliptic} in Section \ref{sec:exit-time}. 
Under $P^x$, the canonical process $X_t(\omega)=\omega(t)$ solves \eqref{eq:ref-sde}.
Define the family of Koopman operators $\{\Kcal_t\}_{t\in\Tset}$ as
\begin{align}
\Kcal_t \varphi(x) = \E[\varphi(X_t)|X_0=x],
\end{align}
per \citep{koopman1932dynamical,lasota2013chaos}. 
Although \eqref{eq:ref-sde} is nonlinear, $\Kcal_t$ for $t\in \Tset$ are \emph{linear} operators and $\{\Kcal_t\}_{t\in\Tset}$ defines a semigroup $\Kcal_{t+s}=\Kcal_t \circ \Kcal_s$. The associated infinitesimal Koopman generator $\Lmd$ is a \emph{linear} second-order differential operator given by
\begin{align}
    \Lmd \varphi(x) = b(x)\cdot\nabla \varphi(x)
        + \frac{1}{2}\sum_{i=1}^{n} \sum_{j=1}^{n} a_{ij}(x)\, \partial_{ij} \varphi(x),
    \label{eq:ref-generator}
\end{align}
for $\varphi \in C^2(\Xset)$, where the indices $i,j = 1,\cdots,n$ range over the coordinates of $x = (x_1,\dots,x_n) \in \Xset$,  $\partial_{ij}$ denotes the second partial derivative with respect to $x_i$ and $x_j$, $a_{ij}$ is the $(i,j)$ entry of the diffusion matrix $a(x):=\sigma(x)\sigma(x)^\top$, and $b$ is the drift term in \eqref{eq:ref-sde}. 
The spectrum of $L$ can be utilized to characterize the global structure of the dynamics \citep{mezic2005spectral,mardt2018vampnets,kostic2022learning}.
Let $(\lambda,\psi)$ be an eigenpair of $\Lmd$, i.e., $\Lmd\psi=\lambda\psi$. Then $\Kcal_t \psi= e^{\lambda t}\psi$, i.e., $\E_x[\psi(X_t)]=e^{\lambda t}\psi(x)$. The constant function $\psi_0\equiv 1$ is always an eigenfunction with $\lambda_0=0$. 

The Koopman operator and its spectrum can be approximately computed from data via numerical methods such as dynamic mode decomposition (DMD) \citep{schmid2010dynamic,rowley2009spectral} and extended dynamic mode decomposition (EDMD) \citep{williams2015data}, as well as nonparametric variants based on RKHS known as kernel EDMD \citep{williams2014kernel,kostic2022learning,klus2020kernel,hou2023sparse}. We provide a brief introduction to RKHS in Appendix \ref{sec:rkhs}. Following \cite{kostic2022learning,hou2023sparse}, we approximate the leading eigenpairs $(\lambda_k,\psi_k)$ of $\Lmd$ from samples $\{x_i\}_{i=1}^N\subset\Xset$. Since $b$ and $a$ are known, $\Lmd$ can be applied directly to the kernel via \eqref{eq:ref-generator}. This yields a generalized eigenproblem in terms of Gram matrices; see, e.g., \cite[Appendix M]{hou2023sparse} for details. 

A trajectory of~\eqref{eq:ref-sde} that crosses from one metastable set to another without returning to its origin is a (reactive) transition path \citep{vanden2006towards}. 
For two disjoint sets $A_{i}$ and $A_{j}$, let $\tau_{A_k}$ be the first time the system enters $A_k$ for $k={i,j}$. The committor function $q^{(i,j)}(x)$ for the transition from $x \in {A_i}$ to ${A_j}$ is defined as $q({x}) = \Pmd \bigl(\tau_{A_j} < \tau_{A_i} \mid {x}_0 = {x}\bigr)$ \citep{vanden2006towards,metzner2009transition,hoffmann2021deeptime}. 
Along a transient path, $q^{(i,j)}(x_t)$ is close to $0$ near $A_i$ and close to $1$ near $A_j$, whereas a path that lingers in $A_i$ has $q^{(i,j)}(x_t) \approx 0$ most of the time. When the metastable sets are clear from context, we write $q$ as a shorthand for $q^{(i,j)}$. Next, let $D_q := \Xset\setminus(A_i\cup A_j)$. The committor function $q$ satisfies the boundary value problem
\begin{align}
    {L} q = 0 \ \text{ in } D_q,
    \qquad q|_{\partial A_i}=0,
    \qquad q|_{\partial A_j}=1,
\label{eq:committor-bvp}
\end{align}
where ${L}$ is the Koopman generator, and we extended $q$ by $q\equiv0$ on $A_i$ and $q\equiv1$ on $A_j$. We defer the proof of this identity to the Appendix~\ref{pf:sec:committor}. 

The above concepts lead to our pipeline for TPS: the leading eigenfunctions of the Koopman generator $\Lmd$ identify the metastable sets and provide the spectral term of the committor function; the committor function then defines the running cost of the OSC presented in the next section, whose optimal controller can be approximated via a boundary value problem in $\Lmd$.
\section{Transition Path Sampling via {OSC} up to an Exit Time}
\label{sec:exit-time}

\subsection{The Sampling Problem}
\label{sec:gibbs}
Consider the diffusion process governed by the It\^o SDE~\eqref{eq:ref-sde}. To facilitate our analysis of the outlined sampling approach, we make the assumption:
\begin{assumption}\label{A1}
 $b(x)$ and $\sigma(x)$ are bounded and Lipschitz on $\Xset$. 
\end{assumption}
Under Assumption 1, the SDE \eqref{eq:ref-sde} has a unique strong solution.
Let $P^x$ be the path measure started at $x$.
Let $D \subset \Xset$ be open, and $\tau_D:=\inf\{t\ge0:X_t\notin D\}$ be the exit time.
Since $D$ is open, $\tau_D$ is the hitting time of the closed set $\Xset \setminus D$ and is therefore an $(\Fcal_t)$-stopping time. We observe paths until they reach the target set (equivalently, exit the set $D$), and we work on the stopped $\sigma$-algebra $\Fcal_{\tau_D}$ defined by $\Fcal_{\tau_D}:=\bigl\{A\in \Fcal:\ A\cap\{\tau_D\le t\}\in\Fcal_t\ \text{for all }t\ge0\bigr\}$. Then $X_{\tau_D}$ is $\Fcal_{\tau_D}$-measurable.

For $x \in \bar D$ and each path $\chi$ under $P^x$, define the Hamiltonian as
\begin{align}
H(\chi) = \int_0^{\tau_D} f(X_t)\,dt + \Phi(X_{\tau_D}),  
\end{align}
where $f:\bar D\to[0,\infty)$ and $\Phi:\partial D\to[0,\infty)$ are measurable, and $H(\chi):=+\infty$ on $\{\tau_D=\infty\}$.
Here, $f \geq 0$ is the running cost and $\Phi$ is the terminal cost. 
Since paths are continuous and $D$ is open, $X_{\tau_D}\in \partial D$ on $\{\tau_D<\infty\}$. Hence, $H$ is a well-defined $\Fcal_{\tau_D}$-measurable function.

For the purpose of efficient sampling, among all path measures on $(\Omega,\Fcal_{\tau_D})$, we seek the one that minimizes the free energy
\begin{align}
  F({\tP}) := \tE \big[H(\chi)\big] 
  + \DKL \big({\tP} \,\|\, P^x\big),
  \qquad {\tP} \in \Pcal(\Omega,\Fcal_{\tau_D}),
  \label{eq:free-energy}
\end{align}
where $\tE$ denotes the expectation w.r.t. $\tP$, and $\Pcal(\Omega,\Fcal_{\tau_D})$ denotes the set of all probability measures on $(\Omega,\Fcal_{\tau_D})$. Here, the first term favors paths with low energy, while the second keeps $\tP$ close to the original dynamics. We next make use of the Gibbs variational principle on $\Fcal_{\tau_D}$, which is adapted from \cite[Proposition 2.1]{mitter2003variational}.
\begin{proposition}
Let $H$ be an $\Fcal_{\tau_D}$-measurable function, and suppose $Z:=\E_{x} \big[e^{-H}\big]>0$. Then the probability measure $P^{*,x}$ on $(\Omega,\Fcal_{\tau_D})$ defined by
\begin{align}
  \frac{\dd P^{*,x}}{\dd P^x}\bigg|_{\Fcal_{\tau_D}}
   =  \frac{e^{-H(\chi)}}{Z}
  \label{eq:gibbs-density}
\end{align}
is the unique minimizer of the free energy \eqref{eq:free-energy} over the set of all probability measures on $(\Omega,\Fcal_{\tau_D})$, and $\min_{{\tP}} F({\tP}) = F(P^{*,x}) = -\log Z$.
\label{prop:gibbs}
\end{proposition}
It will be useful to keep in mind the following motivating example.
\begin{example}[TPS]
{\em Let $A_i,A_j \subset \Xset$ be disjoint closed source and target sets, and let $D:=\Xset\setminus A_j$. Then $A_i\subset D$, and thus a trajectory initialized within the source set starts in the interior of $D$. For such a trajectory, the exit time of $D$ is the first hitting time of $A_j$, so $\tau_D$ is the transition time of interest.  
Let $q(x):=P^x(\tau_{A_j}<\tau_{A_i})$ be the committor function. Define the running and terminal cost
\begin{align}
{f_\beta}(x) := \beta\bigl(1 - q^{(i,j)}(x)\bigr),
\quad
\Phi\equiv 0 \text{ on } \partial A_j,
\end{align}
with $\beta>0$. Then, $H = \beta H_0$, where $H_0=\int_0^{\tau_D}(1-q^{(i,j)}(X_t))\,dt$ is the time spent in nonreactive regions before reaching the target. Here, the parameter $\beta$ controls the trade-off between speed and fidelity. As $\beta$ increases, we increasingly prefer faster paths at the cost of deviating more from the original dynamics. In Appendix~\ref{sec:reweight}  we show how to undo this dynamics distortion via reweighting.}
\label{ex:tps}
\end{example}

The sampling problem of interest is to generate paths from $P^{*,x}$. Proposition~\ref{prop:gibbs} identifies the target measure $P^{*,x}$ but does not provide a way to sample from it. We therefore phrase the sampling problem as a control problem where we add a control drift to the reference process~\eqref{eq:ref-sde} so that the law of the controlled diffusion equals $P^{*,x}$ on $\Fcal_{\tau_D}$. 
We call a control $u:\bar D\to\R^n$ \emph{admissible} if $u$ is Lipschitz continuous on $\bar D$ and use $\Ucal$ to denote the set of all admissible controls. A $u \in \Ucal$ enters the drift through $a(x)u(x)$, resulting in the controlled process
\begin{align}
  d\tilde X_t = \bigl[b(\tilde X_t) + a(\tilde X_t)\,u(\tilde X_t)\bigr]\,dt
  + \sigma(\tilde X_t)\,d\tilde W_t,
  \label{eq:general-controlled-sde}
\end{align}
where $\tilde\tau_D:=\inf\{t\ge0:\tilde X_t\notin D\}$.
As we show in Theorem~\ref{thm:exit-time-SOC} (ii), restricting the focus to a Lipschitz continuous controller does not weaken the result, since we can show that the path measure of the optimal admissible controller is optimal over \emph{all} probability measures on $(\Omega,\Fcal_{\tau_D})$.

As the time-horizon is the first time the process exits $D$, the exit time $\tau_D$ is itself a random variable. We first establish that it has a finite mean. 
We make the following Assumptions
\begin{assumption}\label{A1.elliptic}
There exists $c>0$ such that $z^\top a(x)\,z \geq c|z|^2$ for all $x\in \Xset, z\in\R^n$.
\end{assumption}
\begin{assumption}\label{A2}
  $D\subset \Xset$ is open and bounded with $C^{2,\alpha}$ boundary $\partial D$ for some $\alpha\in(0,1)$.
\end{assumption}
These assumptions are standard for exit-time problems for controlled diffusions \citep{fleming1977exit}. They ensure a unique solution of the associated elliptic boundary value problem in $\Lmd$. They hold, for example, for overdamped Langevin dynamics with a smooth potential, including the double-well, four-well, and M\"uller--Brown systems studied in Section \ref{sec5} and Appendix \ref{sec:experiments}.
The proof of the following lemma is presented in Appendix~\ref{pf:cor:finite-exit-time}:
\begin{lemma}
  Let Assumptions~\ref{A1}--\ref{A2} hold. Then, the SDE~\eqref{eq:ref-sde} starting from $x\in D$ has a pathwise unique strong solution, and its exit time $\tau_D$ satisfies $\E_x[\tau_D] <\infty$. Consequently, $\tau_D<\infty$ a.s.
  \label{cor:finite-exit-time}
\end{lemma}

\subsection{A Candidate Optimal Controller}
\label{sec:fk-ch}
We now seek a controlled diffusion of the form~\eqref{eq:general-controlled-sde} whose path measure equals $P^{*,x}$. For this purpose, we also require a regularity condition, 
\begin{assumption}\label{A3}
  The running cost $f\in C^{0,\alpha}(\bar D)$ with $f \geq 0$, and $\Phi: \partial D \to [0,\infty)$ such that $g:=e^{-\Phi} \in C^{2,\alpha}(\partial D)$.
\end{assumption}

Define
\begin{align}
\rho(x) 
:= \E_x \left[\exp \Bigl(-\int_0^{\tau_D} f(X_t)\,dt - \Phi(X_{\tau_D})\Bigr)\right].
\label{eq:rho-exit}
\end{align}
Here, we integrate time from $0$ to an exit time $\tau_D$. Under Assumption~3, $f$ is bounded on $\bar D$, and $\Phi$ is bounded on $\partial D$ since $g$ is continuous and positive on the compact set $\partial D$. 
By Lemma~\ref{cor:finite-exit-time}, $\tau_D<\infty$ a.s.\ under $P^x$, hence \eqref{eq:rho-exit} is well-defined, and $\rho(x)\in(0,1]$.

We next show that $\rho$ solves a boundary value problem via elliptic existence theory \cite[Ch.~III]{ladyzhenskaya1968linear}.
We consider the following elliptic boundary value problem, where ${L}$ is the (Koopman) generator:
\begin{align}
  {L}w - fw = 0 \text{ in } D, \qquad w|_{\partial D}=e^{-\Phi}.
  \label{eq:elliptic-bvp}
\end{align}
Since $\partial D\in C^{2,\alpha}$ by Assumption~2 and $g=e^{-\Phi}\in C^{2,\alpha}(\partial D)$ by Assumption~3, applying \cite[Ch.~III, Theorem~1.1,~1.2]{ladyzhenskaya1968linear} gives a unique solution $w\in C^{2,\alpha}(\bar D)$ to the problem  \eqref{eq:elliptic-bvp}. 
The proof of the following lemma is included in the Appendix~\ref{pf:lemma.FK}.

\begin{lemma}
Let Assumptions~\ref{A1}--\ref{A3} hold and let $w \in C^{2,\alpha}(\bar D)$ be the unique solution of the problem~\eqref{eq:elliptic-bvp}. Then, $w$ coincides on $\bar D$ with $\rho$ defined in \eqref{eq:rho-exit}. Moreover, $\rho>0$ on $\bar D$.
\label{lemma.FK}
\end{lemma}

Our sampling/free energy minimization problem can now be cast as an optimal control problem. In the fixed-horizon setting, the relative entropy of the controlled law with respect to the reference law equals the expected control effort by Girsanov's theorem, so the free energy~\eqref{eq:free-energy} of the controlled law is an optimal control cost~\citep[eq.~(4.5)]{raginsky2026}. By contrast, here we have to consider the optimal stochastic control problem up to an exit time,
\begin{align}
\begin{aligned}
  &\min_{u\in\Ucal}  J(u,x) :=\tE_x \left[
    \int_0^{\tilde\tau_D} \Bigl(
      f(\tilde X_t) + \tfrac{1}{2}u(\tilde X_t)^\top a(\tilde X_t)\,u(\tilde X_t)
    \Bigr)\,dt
    + \Phi(\tilde X_{\tilde\tau_D})
  \right], \\
  &\text{subject to \eqref{eq:general-controlled-sde}},\ \tilde X_0 = x\in D.
\end{aligned}
  \label{eq:soc}
\end{align}
We denote its value function as $v(x) := \inf_{u\in\Ucal} J(u,x)$, for $x\in \bar D$, with $v|_{\partial D} = \Phi$, since $\tilde\tau_D = 0$ when $\tilde X_0\in\partial D$.
Consider the candidate value function $v:=-\log\rho$ (this is the logarithmic, or Cole--Hopf, transform of $\rho$, \citep{fleming1977exit}). Since $\bar D$ is compact by Assumption~2 and $\rho\in C(\bar D)$ with $\rho>0$ on $\bar D$ as established in Lemma~\ref{lemma.FK}, the extreme value theorem gives $m:=\min_{\bar D}\rho>0$ and $M:=\max_{\bar D}\rho<\infty$, so $\rho(\bar D)\subseteq[m,M] \subset(0,\infty)$. Since $\rho\in C^{2,\alpha}(\bar D)$, and $-\log$ is smooth on $[m,M]$, the chain rule gives $v=-\log\circ\,\rho\in C^{2,\alpha}(\bar D)$, and hence $\grad v$ is bounded on $\bar D$. Additionally, since $\rho\big|_{\partial D}=g=e^{-\Phi}$, the boundary condition of $v$ is $v\big|_{\partial D}=\Phi$.
We have the following lemma whose proof is presented in Appendix~\ref{pf:lemma.V}.

\begin{lemma}
The function $v$ satisfies the Hamilton-Jacobi-Bellman (HJB) equation of the OSC problem~\eqref{eq:soc} up to an exit time, i.e., 
\begin{align}
  Lv(x) + f(x) = \tfrac12\,\nabla v(x)^\top a(x)\,\nabla v(x)
  \quad\text{in } D, \qquad v\big|_{\partial D} = \Phi.
\label{eq:hjb}
\end{align}
\label{lemma.V}
\end{lemma}
We hence have a candidate for the optimal controller of~\eqref{eq:soc} given by
\begin{align}
  u^*(x) = -\grad v(x) = \grad\log\rho(x) = \grad \rho(x) / \rho(x).
  \label{eq:u-star-exit}
\end{align} 
Theorem~\ref{thm:exit-time-SOC} below establishes that its minimal value equals $\min F = -\log Z$ and that its optimally controlled law is $P^{*,x}$.

\subsection{Verification Theorem of OSC up to an Exit Time}
Let ${\tP}^{*,x}$ be the law of the controlled process~\eqref{eq:general-controlled-sde} under $u^*$, starting at $x$.
In the remainder of the section, we establish that $u^*$ in~\eqref{eq:u-star-exit} is the optimal controller. Moreover, using a variational characterization of the optimally controlled path measure, we show that ${\tP}^{*,x}$ minimizes the free energy~\eqref{eq:free-energy} over all path measures on $(\Omega,\Fcal_{\tau_D})$.  
We begin by presenting an explicit formula for the Radon-Nikodym density of $\tP^{*,x}$ relative to $P^x$ on $(\Omega,\Fcal_{\tau_D})$. To that end, we invoke Girsanov's formula. In the classic fixed time-horizon case, one needs a suitable regularity condition, such as Novikov's condition ~\cite[Theorem 8.6.3]{oksendal2013stochastic}. By contrast, in the exit time setting, such a condition might not hold as it requires an exponential moment bound, which is hard to verify for metastable systems. To address this challenge, the proof presented in Appendix~\ref{pf:lemma:girsanov-density} only requires $\tau_D < \infty$ a.s.\,  guaranteed by Lemma~\ref{cor:finite-exit-time}. 
\begin{lemma}
Let $u^*=-\grad v$ be as in~\eqref{eq:u-star-exit}, and let $\tilde X_t$ be a solution to the controlled SDE~\eqref{eq:general-controlled-sde} with $u=u^*$, initiated at $\tilde X_0=x$, and under a probability measure $\tP^{*,x}$ on $(\Omega,\Fcal_{\tau_D})$. 
Then, we have
\begin{align}
  \left.\frac{d\tP^{*,x}}{dP^x}\right|_{\Fcal_{\tau_D}}
  = \frac{e^{-H(\chi)}}{\rho(x)},
  \qquad
  \E_{x} \left[\frac{e^{-H(\chi)}}{\rho(x)}\right] = 1.
  \label{eq:girsanov-density}
\end{align}
\label{lemma:girsanov-density}
\end{lemma}

We now present the verification theorem the proof of which can be found in Appendix~\ref{pf:thm:exit-time-SOC}.
\begin{theorem}
Let Assumptions~\ref{A1}--\ref{A3} hold, and let $u^*$ be the controller in~\eqref{eq:u-star-exit}. Then,
  \begin{enumerate}
    \item[\textnormal{(i)}] [Optimality]
      $v(x) = J(u^*,x) = \min_u J(u,x)$ for all $x\in D$.
    \item[\textnormal{(ii)}] [Variational inequality]
    For every probability measure ${\tP}$ on $(\Omega,\Fcal_{\tau_D})$, 
    \begin{align}
    J(u^*,x) \le \DKL({\tP}\|P^x)+\tilde\E_x[H(\chi)],
    \label{eq:variational-ineq-exit}
    \end{align}
    with equality if and only if $d{\tP}/dP^x=e^{-H(\chi)}/\rho(x)$ on $\Fcal_{\tau_D}$;
    equivalently, if and only if ${\tP}={\tP}^{*,x},$ corresponding to the law of the
    optimally-controlled process restricted to $\Fcal_{\tau_D}$.
  \end{enumerate}
\label{thm:exit-time-SOC}
\end{theorem}
Part~(i) of the above theorem states that the minimal control cost is $\min_{u\in\Ucal} J(u,x) = v(x) = -\log\rho(x)$, and by Part~(ii), that the free energy~\eqref{eq:free-energy} satisfies $F(\tP) \geq v(x)$ for \emph{every} probability measure on $(\Omega,\Fcal_{\tau_D})$, with equality precisely for $\tP = \tP^{*,x}$. Hence,  the two minimal values of~\eqref{eq:free-energy} and~\eqref{eq:soc} coincide, and the optimal path measure that minimizes the free energy is exactly the one generated by the optimal controller.
In the proof of Part (i) of Theorem \ref{thm:exit-time-SOC}, we also show that the minimizer is unique in the following sense: $J(u,x)=v(x)$ for some $u\in\Ucal$ if and only if $u(\tilde X_t)=u^*(\tilde X_t)$ for Lebesgue-a.e. $t\in[0,\tilde\tau_D]$, $\tP^{u,x}$-a.s.

By the definition in~\eqref{eq:rho-exit}, $\rho(x)$ is large where the trajectory of the original dynamics \eqref{eq:ref-sde} tends to exit quickly and small where it dwells in nonreactive regions. The optimal controller in \eqref{eq:u-star-exit} therefore points toward states from which a trajectory is more likely to exit. The normalization by $\rho>0$ in \eqref{eq:u-star-exit} provides a stronger push towards where $\rho$ is small and the original dynamics is least likely to escape; thus the controller acts more aggressively when deep inside the source basin.

Using Lemma \ref{lemma:girsanov-density}, we can recover statistics of the original dynamics \eqref{eq:ref-sde} from the accelerated path alone. We present a reweighting mechanism via Girsanov's change of measure in Appendix \ref{sec:reweight}. 

\section{Koopman-Based Construction of the Optimal Controller}
\label{sec:opt-controller}
Theorem~\ref{thm:exit-time-SOC} reduces the problem of sampling from $P^{*,x}$ to computing the optimal controller $u^*=\grad\log\rho$, where $\rho$ solves the elliptic problem \eqref{eq:elliptic-bvp}. For TPS in Example~\ref{ex:tps}, applying this result requires three ingredients: the source and target sets $A_i,A_j$, the committor function $q$ which defines the running cost $f_\beta$, and $\rho$ itself. Each can be characterized through the Koopman generator $\Lmd$. Specifically, the metastable sets can be computed through the leading eigenfunctions of $\Lmd$, $q$ through $\Lmd q=0$ in \eqref{eq:committor-bvp}, and $\rho$ through $\Lmd\rho-f_\beta\rho=0$ in
\eqref{eq:elliptic-bvp}.

\subsection{Constructing Metastable Cores and Committor Functions}
\label{sec:committor}
The spectrum of the linear operator $\Lmd$ reveals global geometric structures of the underlying dynamics \citep{mezic2005spectral,mauroy2016global}. In molecular dynamics, for example, the metastable sets can be estimated using leading eigenfunctions whose eigenvalues are close to zero, followed by a spectral gap ~\citep{mardt2018vampnets,hou2023sparse,hou2026SICON}. 
Specifically, the leading eigenvalues
$0 = \lambda_0 > \lambda_1 \geq \cdots \geq \lambda_{m-1}  \gg \lambda_{m} \cdots$ also capture the slow-mode, inter-basin dynamics. By truncating at $m$ modes, fast modes are removed as they encode intra-basin dynamics. The leading eigenfunctions $\{\psi_k\}_{k=1}^{m-1}$ are approximately constant within each metastable set but take different values across them; we omit $k=0$ since $\psi_0 \equiv 1$. Let $\Psi(x)=[\psi_1(x),\dots,\psi_{m-1}(x)] \in \Rset^{m-1}$, where $\{x_n\}_{n=1}^N\subset \Xset$ are data samples and the eigenpairs are estimated from them via kernel EDMD in Section \ref{sec:preliminaries}. We then apply $k$-means clustering to $\{\Psi(x_n)\}_{n=1}^N$, which partitions the samples into $m$ clusters with index sets $\Ical_1,\dots,\Ical_m$. For each cluster $k$ and the set of leading eigenfunctions $\psi_l$, for $l=1,\dots,m-1$, we define the core level of the $l$-th eigenfunction as
\begin{align}
\theta^{(k)}_l := \frac{1}{|\Ical_k|}\sum_{n\in\Ical_k}\psi_l(x_n), \qquad \theta^{(k)}:=\big[\theta^{(k)}_1,\dots,\theta^{(k)}_{m-1}\big]^\top,
\label{eq:core-level}
\end{align}
so that $\theta^{(k)}$ is the $k$-means centroid of cluster $k$. 
Let $\varepsilon_k>0$ be a radius.
Since each $\psi_l$ is approximately constant on a metastable set with value approximately $\theta^{(k)}_l$ on the $k$-th one, we define the core of the $k$-th metastable set for $k = 1,\cdots,m$ as
\begin{align}
A_k := \{x\in \Xset : \|\Psi(x)-\theta^{(k)}\|^2 \leq \varepsilon_k^2\}.
\label{eq:core}
\end{align}

Leveraging the global geometric structure carried by the leading eigenfunctions of the Koopman generator, we next approximate the committor function $q$ via an ansatz $q = \tilde\psi + h$, where $\tilde\psi$ is a spectral term constructed from leading eigenfunctions and $h \in \Hcal_\kappa$ is a remainder with $\Hcal_\kappa$ is an RKHS equipped with kernel $\kappa$, since in general, the solution of \eqref{eq:committor-bvp} may not lie in $\text{span}\{\psi_l\}_{l=1}^{m-1}$. The spectral term $\tilde\psi$ contributes to the global shape of the transition, and also captures the slow coordinate along which the transition occurs, without requiring a reaction coordinate to be specified in advance.
Due to space limitations, detailed construction of $\tilde\psi$ that maps the core of $A_i$ to $0$ and the core of $A_j$ to $1$ is presented in Appendix~\ref{supp.committor}. 
We compute $h \in \Hcal_\kappa$ via $h(x) = \sum_{l\in\Ical_q}\alpha_l\,\kappa(x,x_{l})$ with $\{x_l\}_{l\in\Ical_q}\subset D_q$ and $\Ical_q=\{1,\dots,N_q\}$. The coefficients $\alpha$ are obtained from an equality-constrained convex quadratic program that minimizes the residual of~\eqref{eq:committor-bvp} in $D_q$, which is also deferred to in Appendix~\ref{supp.committor}. 
Putting together, we compute the committor function via
\begin{align}
\begin{aligned}
\hat{q}(x) := \tilde\psi(x) +\sum_{l\in\Ical_q}\alpha_l\,\kappa(x,x_{l}),
\end{aligned}
 \label{eq:committor-final}
\end{align}

\subsection{Kernel-Based Approximation of the Optimal Controller}
\label{sec:koopman-exit}
Unlike the fixed-horizon case studied in~\citep{raginsky2026},  where $u^*$ depends on time, the exit-time optimal controller $u^*$ in~\eqref{eq:u-star-exit} is stationary. To compute $u^*$, we focus on approximating $\rho$ defined in~\eqref{eq:rho-exit} so that it solves the elliptic boundary value problem~\eqref{eq:elliptic-bvp}. By definition of $\rho$, and $f_\beta = \beta(1-q^{(i,j)})$, $\Phi \equiv 0$ in Example \ref{ex:tps}, we have
\begin{align}
  \rho(x) = \E_x  \left[e^{-\beta {H_0}}\right],
  \qquad
  {H_0} = \int_0^{\tau_D}\bigl(1-q^{(i,j)}(X_t)\bigr)\,dt,
  \label{eq:W-def}
\end{align}
where ${H_0}$ is the time spent in nonreactive regions before reaching the target. We need to handle two regimes, capturing the transition region and regions within the basins. Directly approximating $\rho$ in an RKHS can be challenging, as the behavior of $\rho$ in the two regimes may be quite different. Near $\partial A_j$, the target is reached almost immediately and $\rho \approx 1$. Within the source basin, however, the transition is a rare event, and $\rho$ is small as the trajectory lingers there. Since $u^*$ depends on $\rho$ through $\log \rho$, we seek an ansatz of the form $\rho:= \rho_0(x) w(x)$ so that $\log \rho = \log \rho_0 + \log w$. Here we choose $\rho_0$ in closed form to capture how $\rho$ varies across the two regimes, and use a residual ${w}$ to solve for $\Hcal_\kappa$. 
We defer additional details to Appendix~\ref{sec:supp.rho} and only mention here that we approximate $\rho$ via a closed-form factor $\hat{\rho}$ as
\begin{align}
    \hat{\rho}(x) := \rho_0(x)\,w(x),
    \quad
    \rho_0(x):= \hat q(x)+ \bigl(1-\hat q(x)\bigr) r,\quad
    w \in \Hcal_\kappa,
    \label{eq:rho-kernel-expansion}
\end{align}
where $\Hcal_\kappa$ is an RKHS with kernel $\kappa$, $\hat{q}$ as in~\eqref{eq:committor-final} and $r = \frac{1}{1+\beta T_0}$, with $T_0 >0$ which encodes the time-scale of the original dynamics; for example, we set $T_0 = \frac{1}{|\lambda_1|},$ where $\lambda_1$ is the leading nontrivial Koopman eigenvalue. 
Let $\{x_i\}_{i=1}^N = \{x_i\}_{i\in I_D}\cup\{x_i\}_{i\in I_{\partial D}}$, with interior points $\{x_i\}_{i\in I_D} \subset D$ and boundary points $\{x_i\}_{i\in I_{\partial D}} \subset \partial D$. We approximate $w \in \Hcal_\kappa$ via $w= \sum_{i=1}^N c_i\,\kappa(x,x_i)$ (Appendix~\ref{sec:supp.rho}). Applying the chain rule to $\log\hat{\rho}=\log\rho_0+\log w$, we have
\begin{align}
  \hat{u}(x)
  = \nabla \log \hat{\rho}(x)
  = \frac{(1-r)\,\nabla \hat q(x)}{\hat q(x) + \bigl(1-\hat q(x)\bigr) r}
  + \frac{\nabla w(x)}{w(x)}.
  \label{eq:uK}
\end{align}
Thus, the problem of solving for $u^*$ reduces to solving a convex quadratic program with linear equality constraints, whose solution is characterized by a single linear KKT system as in Appendix~\ref{sec:supp.rho}.
The method is summarized in Algorithm~\ref{sec:algorithm-et}.

\begin{algorithm}[ht]
\caption{Summary of the Algorithm in Four Stages}
\label{sec:algorithm-et}
\begin{algorithmic}[0]
  \State \textbf{Stage I.} Construct leading eigenpairs $(\lambda_k,\psi_k)_{k=1}^{m-1}$ of the Koopman generator $\Lmd$ and identify metastable cores.
  \State \textbf{Stage II.} Compute the committor function $\hat{q}$ via~\eqref{eq:committor-final} using the leading Koopman eigenfunctions.
  \State \textbf{Stage III.} Construct the optimal controller $\hat{u}$ via~\eqref{eq:uK}.
  \State \textbf{Stage IV.} Sample starting points from the source set $A_i$ and integrate the controlled dynamics~\eqref{eq:general-controlled-sde} with $\hat{u}$ from \eqref{eq:uK} until the paths hit the target set.
\end{algorithmic}
\end{algorithm}
\section{Numerical Experiments}
\label{sec5}
We now present empirical evaluations of the proposed method on two synthetic systems, namely, the two-channel double-well system (Figure~\ref{fig:2ch-paths}) and the M\"uller-Brown system (Figure~\ref{fig:muller_brown_paths}), and a real-world alanine dipeptide (Figure~\ref{fig:alanine_sample_paths}). 
Due to space limitations, experimental details and one additional example are deferred to Appendix~\ref{sec:experiments}. In all these experiments, a path is a hit at its first entrance to the target, and is truncated at the numerical censor $T_{\max}$ otherwise. We report the target-hit percentage (THP), which is the fraction of paths that hit the target set before $T_{\mathrm{max}}$. We also remark that alanine dipeptide does not satisfy Assumption \ref{A1.elliptic}; nevertheless, our algorithm can be applied beyond the confines of the theoretical analysis, and this experiment serves as a demonstration that it performs well even in such cases. 

\begin{figure}[t]
  \centering
  \includegraphics[width=0.9\linewidth]{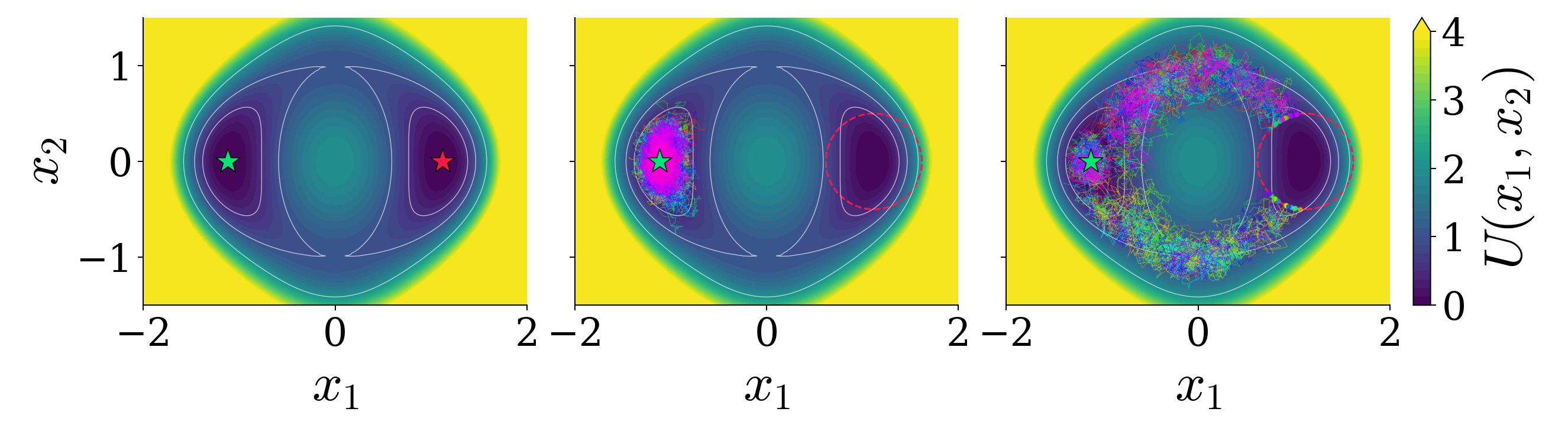}
  \caption{
  Two-channel double well. The potential has two 
  minima at $(\pm\sqrt{5}/2,\,0)$ connected by two channels around $(0,\pm 1)$. Following Algorithm~\ref{sec:algorithm-et}, we compute the Koopman eigenfunctions and extract the metastable cores $A_{\mathrm{left}}$ and $A_{\mathrm{right}}$ as described in Section~\ref{sec:committor}. Their minima are marked by the green and red stars, and the dashed red curve marks $\partial A_{\mathrm{right}}$. We sample $1000$ uncontrolled and $1000$ controlled paths, all started at the minimum of $A_{\mathrm{left}}$. (Left) Potential energy landscape. (Middle) Uncontrolled paths: THP($T_{\mathrm{max}}$) = $0\%$. (Right) Proposed method: THP($T_{\mathrm{max}}$) = $99.8\%$.}
  \label{fig:2ch-paths}
\end{figure}

\begin{figure}[t]
    \centering
    \begin{subfigure}[b]{0.31\textwidth}
    \centering
        \includegraphics[height = 3.2cm]{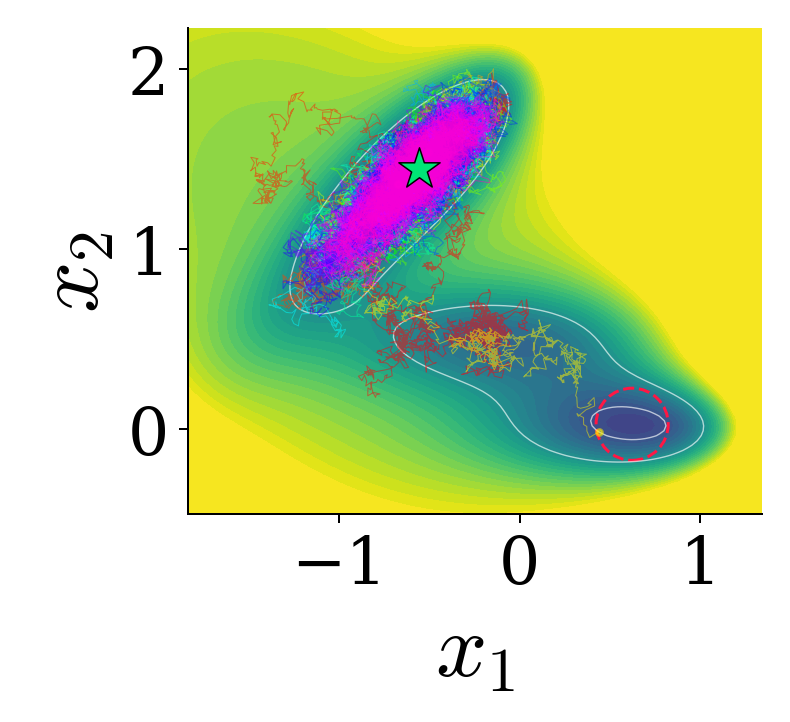}
        \caption{Uncontrolled}
        \label{fig:sub_a_mb_paths}
    \end{subfigure}
    \hfill
    \begin{subfigure}[b]{0.31\textwidth}
    \centering
        \includegraphics[height = 3.2cm]{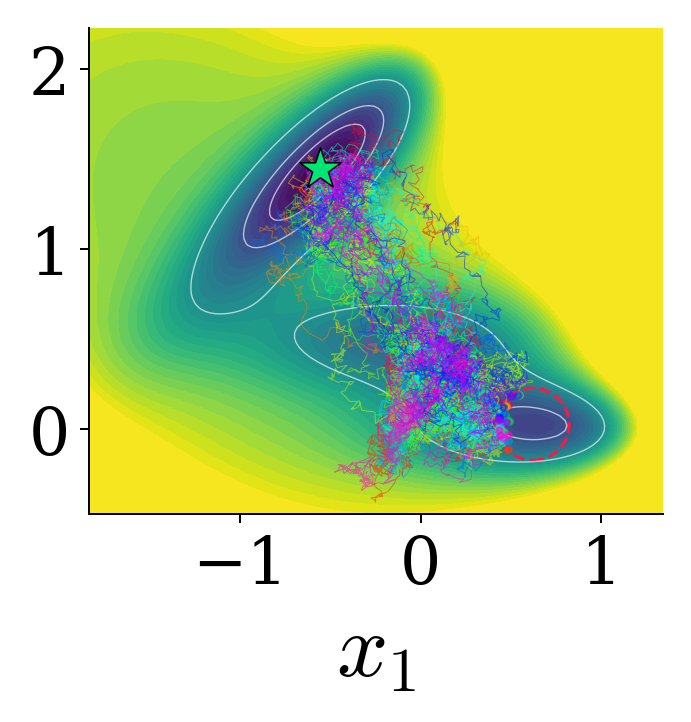}
        \caption{TPS-DPS}
        \label{fig:sub_b_mb_paths}
    \end{subfigure}
    \hfill
    \begin{subfigure}[b]{0.36\textwidth}
    \centering
        \includegraphics[height = 3.2cm]{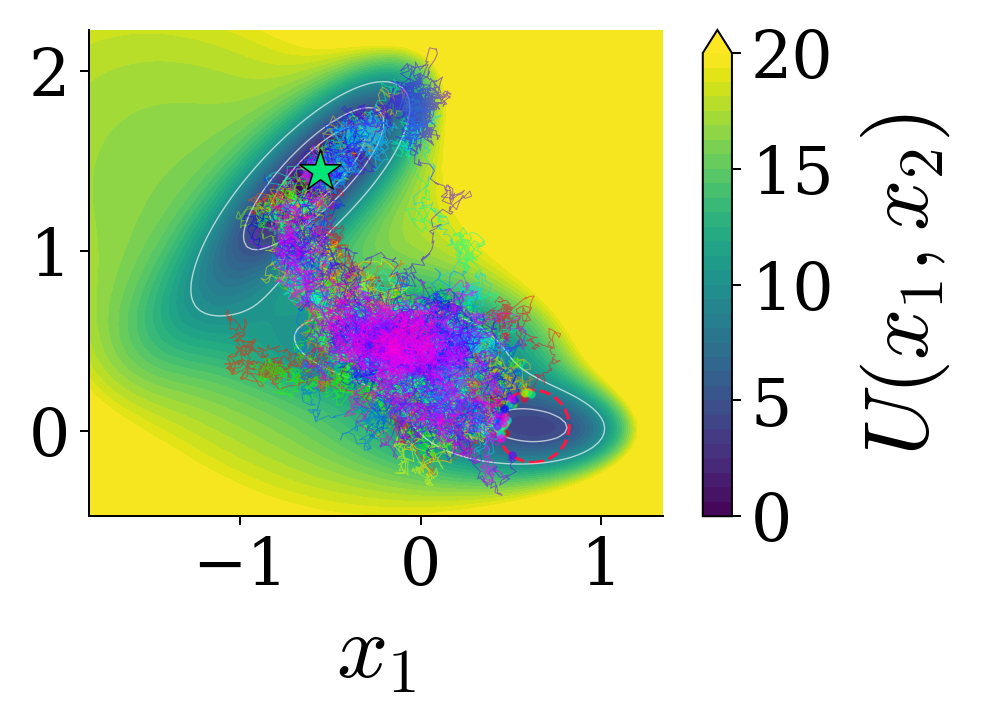}
        \caption{Proposed method}
        \label{fig:sub_c_mb_paths}
    \end{subfigure}
    \caption{M\"uller-Brown potential, with a deep well at $(-0.558, 1.442)$ (green star), a medium well at $(0.623, 0.028)$ (red circle), and a shallow well at $(-0.050, 0.467)$. We aim to sample transitions from the deep well to the medium well within $T_\mathrm{max} = 10\mathrm{ps}$.
    (a) Uncontrolled paths: THP($T_{\mathrm{max}}$) = $1.8\%$. (b) TPS-DPS~\citep{seong2024tpsdps}: THP($T_\mathrm{max}$) = $100\%$, but the paths run over the high energy barrier. (c) Proposed method: THP($T_\mathrm{max}$) = $99.0\%$ using $\beta = 1$, with paths following the low-energy channel through the shallow well rather than crossing the barrier.}
    \label{fig:muller_brown_paths}
\end{figure}

\begin{figure}[t]
    \centering
    \begin{subfigure}[b]{0.3\textwidth}
        \centering
        \includegraphics[height=3cm]{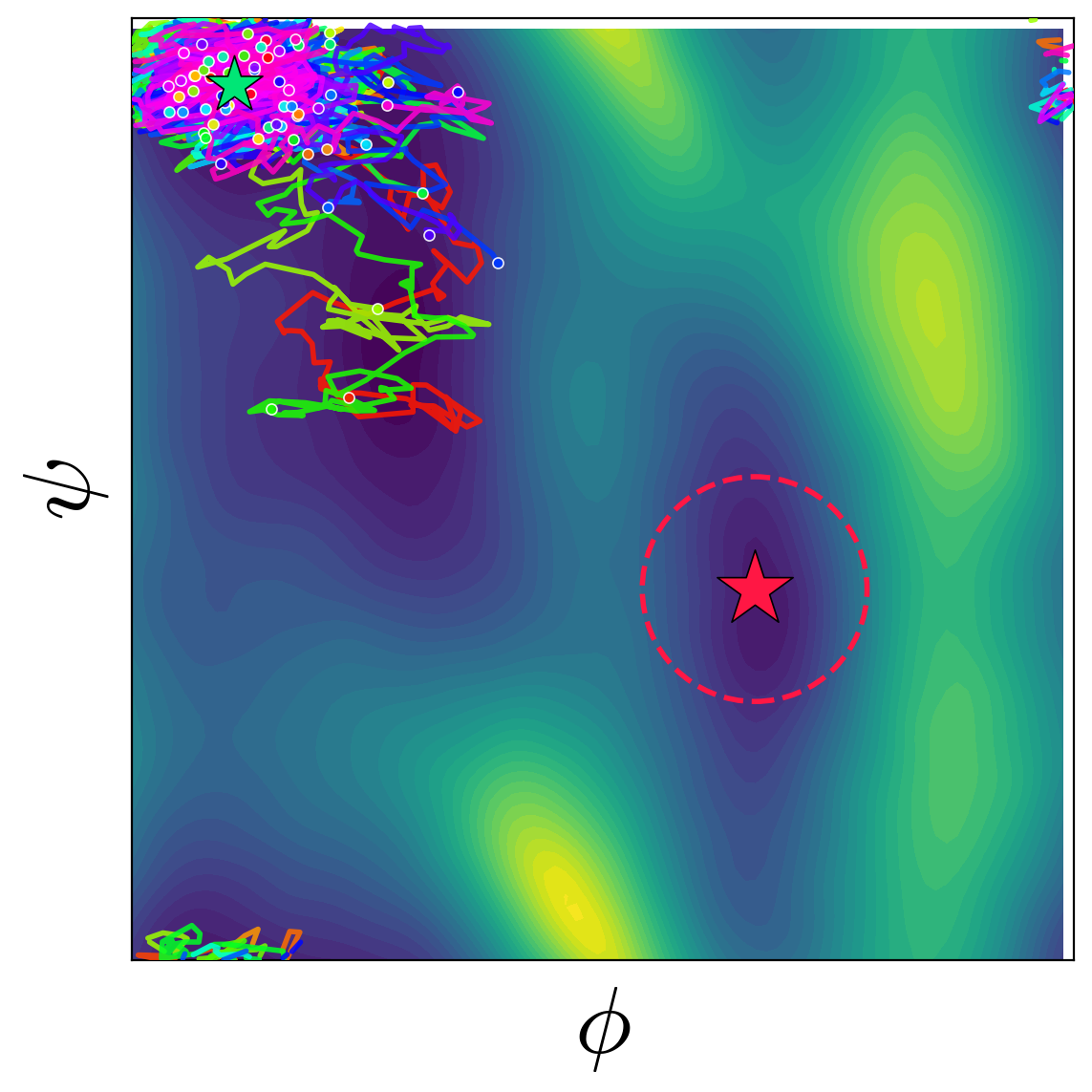}
        \caption{Uncontrolled}
        \label{fig:sub_a}
    \end{subfigure}
    \hfill
    \begin{subfigure}[b]{0.3\textwidth}
        \centering
        \includegraphics[height=3cm]{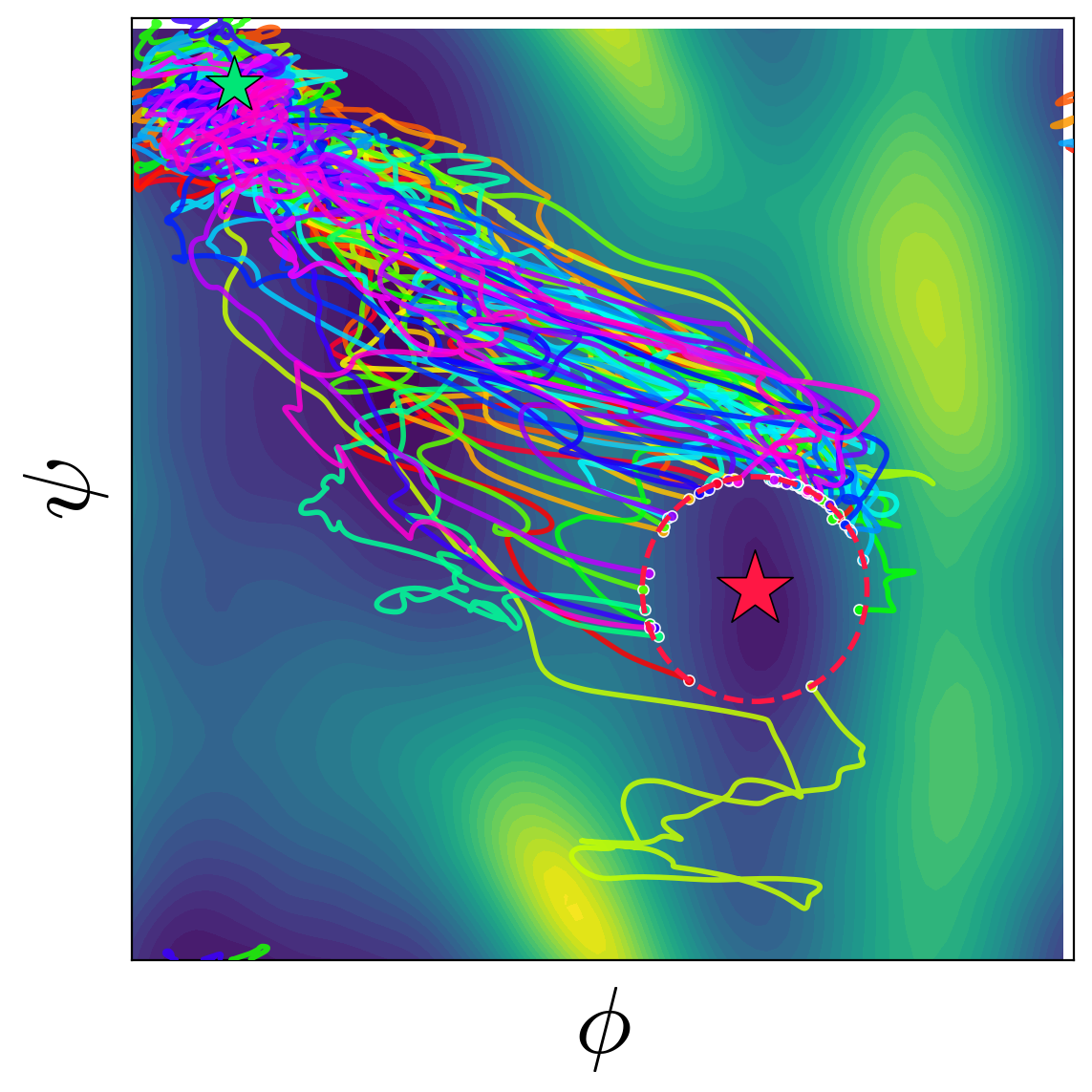}
        \caption{TPS-DPS}
        \label{fig:sub_a}
    \end{subfigure}
    \hfill
    \begin{subfigure}[b]{0.35\textwidth}
        \centering
        \includegraphics[height=3cm]{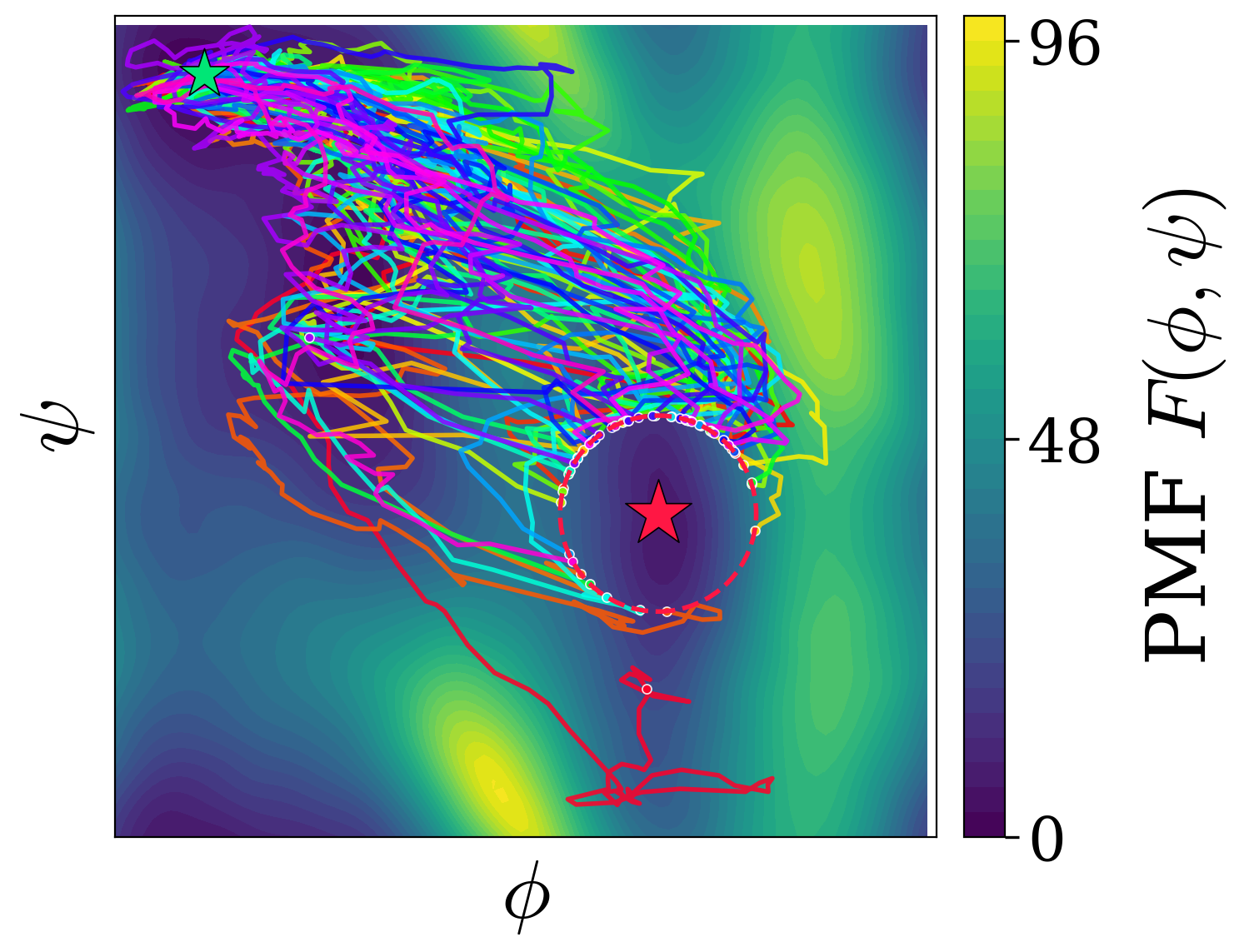}
        \caption{Proposed method}
        \label{fig:sub_b}
    \end{subfigure}
    \caption{Alanine dipeptide in vacuum, shown on the free-energy surface over the backbone dihedral angles $(\phi, \psi)$. We aim to sample the transition from $C5$ conformation (upper left) to the $C7_{\mathrm{ax}}$ conformation (red star) within $T_{\mathrm{max}}=1\mathrm{ps}$. (a) Uncontrolled paths: $(T_{\mathrm{max}})= 0\%$ (b) TPS-DPS: THP $(T_{\mathrm{max}})= 100\%$, obtained by training a neural-network bias with simulation in the loop ($1000$ rollouts, ${\sim}10^6$ gradient updates). (c) Proposed method: THP $(T_{\mathrm{max}})= 93\%$, with a controller obtained in closed form. The setup and computational cost are reported in Appendix~\ref{sec:experiments}.}
    \label{fig:alanine_sample_paths}
\end{figure}

\bibliographystyle{plainnat} 
\bibliography{diffusion}

\appendix
\setcounter{section}{0}
\clearpage 
\section*{\Large \centering Appendix}
\section{The Reproducing Kernel Hilbert Space}
\label{sec:rkhs}
We provide a brief introduction to the real-valued RKHS -- for a more comprehensive treatment, see~\citep{berlinet2011reproducing}. Let $\Xset$ be a subset of a Euclidean space and let $\kappa: \Xset \times \Xset \rightarrow \mathbb{R}$ be a continuous, symmetric, positive semi-definite kernel. Define the RKHS associated with the kernel $\kappa$, $\Hcal_\kappa,$ as the completion of the span of  $\left\{\phi(x):=\kappa(x, \cdot) : x \in \Xset \right\}$, equipped with the inner product $\langle\cdot, \cdot\rangle$, satisfying $\left\langle \phi({x}), \phi({y})\right\rangle=\kappa(x, y)$. Here, $\phi$ is called the feature map for the kernel $\kappa$. The inner product satisfies the \emph{reproducing property}, given by
$\left\langle \phi(x), f\right\rangle=f(x)$, for all $x \in \Xset$ and $f \in \Hcal_\kappa$. Throughout the exposition, we make the standard assumption that $\kappa$ is measurable.

\subsection{Auxiliary Result: A Localized Dynkin Formula}
\label{sec:localized-dynkin}
Several proofs in this paper follow a similar outline: we start with a real-valued function ${h}$ that solves a partial differential equation on $D$, and we aim to evaluate the expectation of ${h}$ along the trajectory at the exit time $\tau_D$. However, It\^o's formula is stated for functions that are $C^2$ on $\Rset^n$, whereas the functions we work with are $C^2$ only on the set $D$. Therefore, we need a mechanism to integrate up to $\tau_D$ at which time the trajectory reaches $\partial D$. In addition, such a $\tau_D$ is a random variable itself. We address this by localizing, following the scheme in the proof of \cite[Theorem 9.1.1]{oksendal2013stochastic}. Specifically, we first integrate to $\sigma_k:= k \wedge \tau_{D_k}$, where $D_k$ is an increasing sequence of open sets in $D$. Then on $[0,\sigma_k]$, the optional stopping theorem \cite[Theorem 1.3.22]{karatzas2014brownian} applies. Letting $k \to \infty$ then takes the result to $\tau_D$, given $\sigma_k \to \tau_D$. The following lemma (localized Dynkin formula) makes this construction precise.

\begin{lemma}
Let Assumption~\ref{A1}--\ref{A3} hold. Let $Y$ be the pathwise unique strong solution of
\begin{align}
  dY_t = {b}(Y_t)\,dt + \sigma(Y_t)\,dW_t, \qquad Y_0 = x \in D,
  \label{eq:generic-sde}
\end{align}
with generator ${L_Y}{h} := {b}^\top\grad{h} + \tfrac12\sum_{ij}a_{ij}\partial_{ij}{h}$, and let $\tau_D := \inf\{t>0: Y_t \notin D\}$. 
Define a sequence of open sets as
\begin{align}
  D_k := \{y \in D : \operatorname{dist}(y,\partial D) > 1/k\}.
  \label{eq:localizing-seq}
\end{align}
Let $\sigma_k := k \wedge \tau_{D_k}$, where $\tau_{D_k}$ is the exit time of $Y$ from $D_k$. Then, we have
\begin{enumerate}
  \item[\textnormal{(i)}] The sequence $(\sigma_k)_{k\ge1}$ is nondecreasing, and $\sigma_k \to \tau_D$ as $k\to\infty$.
  \item[\textnormal{(ii)}] Let ${h} \in C^2(D)\cap C(\bar D)$, 
  ${c} \in C(\bar D)$ with ${c} \ge 0$, and $Z_t := \exp\bigl(-\int_0^t {c}(Y_s)\,ds\bigr)$. Then for every $k$, we have
        \begin{align}
          {h}(Y_{\sigma_k})Z_{\sigma_k} - {h}(x)
          = \int_0^{\sigma_k} Z_t\bigl({L_Y}{h} - {c}{h}\bigr)(Y_t)\,dt
            + \int_0^{\sigma_k} Z_t\,\grad{h}(Y_t)^\top\sigma(Y_t)\,dW_t,
          \label{eq:localized-ito}
        \end{align}
        and
        \begin{align}
          \E\bigl[{h}(Y_{\sigma_k})Z_{\sigma_k}\bigr] - {h}(x)
          = \E\Bigl[\int_0^{\sigma_k} Z_t\bigl({L_Y}{h} 
          - {c}{h}\bigr)(Y_t)\,dt\Bigr].
          \label{eq:localized-dynkin}
        \end{align}
\end{enumerate}
\label{lem:localized-dynkin}
\end{lemma}

\begin{proof}
By construction of $D_k$ in \eqref{eq:localizing-seq},  $D_k\subset D_{k+1}$, and $\bigcup_k D_k=D$. Moreover, $\bar D_k$ is a compact subset of $D$.

\emph{Part (i)} For every $k \geq 1$, since $D_k\subset D_{k+1}$, we have  $\{t>0:Y_t\notin D_{k+1}\}\subseteq\{t>0:Y_t\notin D_k\}$. Hence, $\tau_{D_k}\leq\tau_{D_{k+1}}$. Since $\sigma_k=k\wedge\tau_{D_k}$ is the minimum of two nondecreasing sequences, it is itself also nondecreasing.

We next show that $\sup_k\sigma_k= \tau_D$ pathwise. Since $D_k \subset D$, we have $\sigma_k \le \tau_D$ for every $k$, so $\sup_k \sigma_k \le \tau_D$. For the reverse direction, fix $\omega$ and $t < \tau_D$. Then $\{Y_s : 0 \le s \le t\}$ is a compact subset of $D$ and hence has strictly positive distance to $\partial D$, so it is contained in $D_{k_0}$ for some $k_0$. For $k>\max(k_0,t)$ we have $D_{k_0}\subseteq D_k$, so $\tau_{D_k}>t$ and $\sigma_k=k\wedge\tau_{D_k}>t$. As $t < \tau_D$ was arbitrary, we have $\sup_k \sigma_k \ge \tau_D$.

\emph{Part (ii)} Fix $k \geq 1$. We distinguish two cases, depending on whether the initial point $x$ is in $\bar D_k$ or not. For the first case, consider $x\notin\bar D_k$. Since $\Xset\setminus\bar D_k$ is open and paths are continuous, there is $\epsilon>0$ such that $Y_t\notin\bar D_k$, and in particular $Y_t\notin D_k$, for all $t\in(0,\epsilon)$. Hence $\tau_{D_k}=0$ and $\sigma_k=0$. Then $t\wedge\sigma_k=0$ for every $t\ge0$. We have $h(Y_0)Z_0-h(x)=0$ as $Z_0=1$ and all integrals from $0$ to $\sigma_k = 0$ are zero. Hence, both sides of \eqref{eq:localized-ito} and \eqref{eq:localized-dynkin} vanish.

For the second case where $x\in\bar D_k$, we first show that $Y_t\in\bar D_k$ for all $t\in[0,\sigma_k]$. Indeed, $Y_0=x\in\bar D_k$, and for $t\in(0,\sigma_k)$, we have $t<\tau_{D_k}$, hence $Y_t\in D_k$. If $\sigma_k>0$, then $Y_{\sigma_k}=\lim_{t\to\sigma_k}Y_t\in\bar D_k$ by path continuity.

Since $\bar D_k$ is a compact subset of the open set $D$, there exists some $\zeta_k \in C_c^\infty(D)$ with $\zeta_k \equiv 1$ on an open neighbourhood $U_k$ of $\bar D_k$, so that in particular, $U_k\subseteq\operatorname{supp}\zeta_k\subset D$. Let ${h}_k$ equal $\zeta_k{h}$ on $D$ and $h_k:=0$ on $\Rset^n \setminus D$. The open sets $D$ and $\Rset^n\setminus\operatorname{supp}\zeta_k$ cover $\Rset^n$, $h_k$ is $C^2$ on the first and identically zero on the second, and the two definitions agree on the overlap. Thus $h_k\in C^2(\Rset^n)$ with compact support. In addition, we have ${h}_k = {h}$, $\grad{h}_k = \grad{h}$ and ${L_Y}{h}_k = {L_Y}{h}$ on $U_k$.

Since $Z_t=\exp(-\int_0^t c(Y_s)\,ds)$ is absolutely continuous in $t$ with $dZ_t=-Z_t\,c(Y_t)\,dt$, it has locally bounded variation and hence $dZ_t^\top d Y_t=0$. Applying It\^o's product rule to ${h}_k(Y_t)Z_t$ on $[0,t\wedge\sigma_k]$, we have 
\begin{align}
 h_k(Y_{t\wedge\sigma_k})Z_{t\wedge\sigma_k}-h_k(x)
  =\int_0^{t\wedge\sigma_k} Z_s\bigl(L_Yh_k-c\,h_k\bigr)(Y_s) ds
   +\int_0^{t\wedge \sigma_k} Z_s \grad h_k(Y_s)^\top\sigma(Y_s) dW_s.   
\end{align}
Since $Y_s\in\bar D_k \subset U_k$ for $s \leq t \wedge \sigma_k$, every term above may be written in terms of ${h}$, yielding \eqref{eq:localized-ito}.

Finally, $\sigma_k \le k$ is a bounded stopping time, $\grad{h}_k$ is continuous with compact support, hence bounded on $\Rset^n$. In addition, $\sigma$ is bounded by Assumption \ref{A1.elliptic}, and $Z_s \in (0,1]$ since ${c} \ge 0$. Hence, we have $\E\bigl[\int_0^k|Z_s\,\grad{h}_k(Y_s)^\top\sigma(Y_s)|^2 ds\bigr] < \infty$ and the optional stopping theorem \cite[Theorem 1.3.22]{karatzas2014brownian} applies.  Taking $t=k$ in \eqref{eq:localized-ito}, so that $t\wedge\sigma_k=\sigma_k$, and taking expectations proves Eq.~\eqref{eq:localized-dynkin}.
\end{proof}

\section{Proof of Results in Section \ref{sec:exit-time}}
\subsection{Proof of Lemma \ref{cor:finite-exit-time}}
\label{pf:cor:finite-exit-time}
Since $\bar D$ is bounded by Assumption~\ref{A2}, we have $ b\in C^{0,\alpha}(\bar D)$ for $\alpha\in(0,1)$.
We show that the SDE \eqref{eq:ref-sde} with $  X_0=x\in D$ has a pathwise unique strong solution, and its exit time $ \tau_D$ satisfies $\E_x[ \tau_D]\le\varphi(x)<\infty$, where ${\varphi} \in C^{2,\alpha}(\bar D)$, ${\varphi}\ge 0$ is the unique solution of the elliptic BVP
  \begin{align}
      L{\varphi}=-1 \text{ in } D,\qquad {\varphi}|_{\partial D}=0.
  \end{align}

\emph{Step 1:}
Consider the elliptic BVP $  L{\varphi}=-1$ in $D$, ${\varphi}|_{\partial D}=0$. Since $  b\in C^{0,\alpha}(\bar D)$ and $a=\sigma\sigma^\top\in C^{0,\alpha}(\bar D)$ by Assumption ~\ref{A1.elliptic}, $\partial D\in C^{2,\alpha}$ by Assumption~\ref{A2}, and the boundary data $0\in C^{2,\alpha}(\partial D)$, elliptic regularity theory \citep[Ch.~III, Theorem~1.1 and Theorem~1.2]{ladyzhenskaya1968linear} guarantees existence and uniqueness of a solution ${\varphi} \in C^{2,\alpha}(\bar D)$.

\emph{Step 2:} To establish that ${\varphi}\ge 0$ on $\bar D$, we will show that the minimum of $\varphi$ over $\bar D$ is attained at a boundary point and is therefore equal to $0$ since $\varphi|_{\partial D} \equiv 0$. Suppose, for contradiction, that ${\varphi}$ attains a local minimum at some interior point $x_0\in D$. Then $\grad {\varphi}(x_0)=0$ and $\mathrm{Hess}\left({\varphi}(x_0)\right)\succeq 0$ by the second-order necessary condition for a local minimum. And we have
\begin{align}
    L{\varphi}(x_0)
  = \underbrace{  b(x_0)^\top\grad {\varphi}(x_0)}_{=0}
+\tfrac12\,\mathrm{tr}\bigl(a(x_0)\,\mathrm{Hess}\left({\varphi}(x_0)\right)\bigr)
  \;\ge\; 0.
\end{align}
Here, by uniform ellipticity in Assumption \ref{A1.elliptic}, we have $a(x_0)\succ 0$, and since $\mathrm{Hess}\left({\varphi}(x_0)\right)\succeq 0$, the inequality holds. This contradicts $  L{\varphi}\equiv -1<0$ on $D$. Hence ${\varphi}$ has no interior local minimum, so $\min_{\bar D}{\varphi}=\min_{\partial D}{\varphi}=0$, i.e.\ ${\varphi}\ge 0$ on $\bar D$.

\emph{Step 3:}
Since $b,\sigma$ are bounded and Lipschitz by Assumption~\ref{A1}, the SDE \eqref{eq:ref-sde} has a pathwise unique strong solution for all $t\ge0$ \citep[Thm.~5.2.1]{oksendal2013stochastic}.

\emph{Step 4:}
We next show $\E_x[ \tau_D]\le {\varphi}(x)$. 
By the previous step, Eq.~\eqref{eq:ref-sde} has a pathwise unique strong solution, so Lemma~\ref{lem:localized-dynkin} applies with ${h} = {\varphi}$ and ${c} \equiv 0$. Let $(\sigma_k)_{k \ge 1}$ and $(D_k)_{k \ge 1}$ be the stopping times and open sets defined in Lemma~\ref{lem:localized-dynkin}. Since $L{\varphi} = -1$ on $D$, Eq.~\eqref{eq:localized-dynkin} gives $\E_x[{\varphi}(X_{\sigma_k})] - {\varphi}(x) = -\E_x[\sigma_k]$. 
By Step 2, ${\varphi} \ge 0$ on $\bar D$. Since $D_k\subset D$, we have $\sigma_k\le\tau_{D_k}\le\tau_D$. And, since the path stays in $\bar D$ up to its exit time $\tau_D$, we have $X_{\sigma_k} \in \bar D$. 
Thus we have $\E_x[{\varphi}(X_{\sigma_k})] \geq 0$, and hence
\begin{align}
  \E_x[\sigma_k] \le {\varphi}(x), \qquad k \in \Nset.
  \label{eq:sigma-k-bound}
\end{align}
Since the sequence $(\sigma_k)_{k\ge1}$ is nondecreasing and $\sigma_k \to \tau_D$ from Lemma~\ref{lem:localized-dynkin}(i), monotone convergence applied to \eqref{eq:sigma-k-bound} gives $\E_x[\tau_D] = \lim_k \E_x[\sigma_k] \leq {\varphi} (x) < \infty$.
Here, since ${\varphi} \in C(\bar D)$ and $\bar D$ is compact, $\max_{\bar D}{\varphi}<\infty$, so $\sup_{x\in D}\E_x[ \tau_D]\leq \max_{\bar D} {\varphi} <\infty$. 
Finally, a nonnegative random variable with finite expectation is finite almost surely, hence $ \tau_D<\infty$ a.s.

\subsection{Finite Mean Exit Time for Controlled Processes}
By Assumption~\ref{A1}, $b$ is bounded and Lipschitz on $\Xset$. More generally, for $u\in\Ucal$, $b+au$ is bounded and Lipschitz on $\Xset$, as $a$ is bounded and Lipschitz by Assumption~\ref{A1.elliptic}, and $u$ is Lipschitz on $\bar D$, which is compact, hence bounded, and is extended to $\Xset$. Hence, the following result on the finite mean exit time for every admissible control is immediate by applying Lemma~\ref{cor:finite-exit-time} to $\tilde b=b+au$.

\begin{corollary}
Consider the controlled SDE \eqref{eq:general-controlled-sde}, and denote $\tilde{L} g:=(b+au)^\top \grad g+\tfrac12\sum_{ij}a_{ij}\partial_{ij}g$. 
Let Assumptions~\ref{A1}--\ref{A2} hold and let $u\in\Ucal$.
Then, the controlled SDE~\eqref{eq:general-controlled-sde} starting from $x \in D$ has a pathwise unique strong solution $\tilde X$, and its exit-time $\tilde\tau_D$ satisfies $\tilde\E_x[\tilde\tau_D] <\infty$. Thus we also have $\tilde\tau_D<\infty$, ${\tP}^{u,x}$-a.s.\ 
\label{cor:finite-exit-time-control}
\end{corollary}

\subsection{Proof of Lemma \ref{lemma.FK}}
\label{pf:lemma.FK}

To show $w = \rho$ on $D$, let $Y_t:=w(X_t)$ and $Z_t := \exp(-\int_0^t f(X_s)\,ds)$. Fix $x \in D$ and we apply Lemma~\ref{lem:localized-dynkin} with ${h} = w$, and ${c}=f$ that satisfies Assumption~\ref{A3}. Since $Lw - fw = 0$ in $D$, the drift term in \eqref{eq:localized-dynkin} vanishes and we have
\begin{align}
  \E_x\bigl[w(X_{\sigma_k})Z_{\sigma_k}\bigr] = w(x), \qquad k \in \Nset.
  \label{eq:w-sigma-k}
\end{align}

We now take $k\to\infty$. By Lemma~\ref{cor:finite-exit-time}, $\tau_D<\infty$, $P^x$-a.s. Since $\sigma_k\to\tau_D$ and $D$ is open, $X_{\sigma_k}\to X_{\tau_D}\in\partial D$.
By continuity of $w$ on $\bar D$ and using the boundary condition \eqref{eq:elliptic-bvp}, we have $Y_{\sigma_k}=w(X_{\sigma_k})\to e^{-\Phi(X_{\tau_D})}$ a.s. By continuity of $t\mapsto\int_0^t f(X_s)\,ds$, we have $Z_{\sigma_k}\to\exp(-\int_0^{\tau_D}f(X_s)\,ds)$ a.s.
Moreover, since $w\in C(\bar D)$ with $\bar D$ compact and $Z_t\in(0,1]$, we have $|Y_{\sigma_k}Z_{\sigma_k}|\le\max_{\bar
D}|w|<\infty$ for every $k$. Hence, $Y_{\sigma_k}Z_{\sigma_k}\to e^{-\Phi(X_{\tau_D})}\exp(-\int_0^{\tau_D}f\,ds)$
a.s., and this sequence is uniformly bounded by $\max_{\bar D}|w|$. Thus, the bounded convergence theorem gives
\begin{align}
\lim_{k\to\infty}\E_x\bigl[Y_{\sigma_k}Z_{\sigma_k}\bigr]
  = \E_x \left[e^{-\Phi(X_{\tau_D})}\exp \left(-\int_0^{\tau_D}f(X_s)\,ds\right)\right]
  = \rho(x),
\end{align}
where the last equality follows from the definition $\rho$ in \eqref{eq:rho-exit}. On the other hand, $\E_x[Y_{\sigma_k}Z_{\sigma_k}]=w(x)$ for every $k$ by
\eqref{eq:w-sigma-k}, so the limit on the left is $w(x)$. Therefore, we conclude $w(x)=\rho(x)$ on $D$.
For $x \in \partial D$, we have $\tau_D = 0$, $P^x$ a.s.\, hence \eqref{eq:rho-exit} gives $\rho(x) = e^{-\Phi(x)} = w(x)$ via the boundary condition \eqref{eq:elliptic-bvp}. Hence, $w(x)=\rho(x)$ on $\bar D$.

It remains to show $\rho>0$ on $\bar D$. On the boundary, $\rho|_{\partial D}=e^{-\Phi}=g>0$ by Assumption~\ref{A3}. In the interior, fix $x\in D$: by Assumption~\ref{A3}, $f\ge0$ is bounded on $\bar D$ and $\Phi=-\log g$ is bounded on $\partial D$, while $\tau_D<\infty$ a.s.\ by Lemma~\ref{cor:finite-exit-time}. Hence, we have
\begin{align}
  \exp \Bigl(-\int_0^{\tau_D}f(X_t)\,dt-\Phi(X_{\tau_D})\Bigr)>0,
  \qquad P^x\text{-a.s.},
\end{align}
and taking expectations gives $\rho(x)>0$. Combining these, $\rho>0$ on $\bar D$. Since $\rho\in C(\bar D)$ and $\bar D$ is compact, it follows that $\min_{\bar D}\rho>0$.

\subsection{Proof of Lemma \ref{lemma.V}}
\label{pf:lemma.V}
Recall $x = (x_1,\dots,x_n) \in \Xset$, and $\partial_{ij}$ denotes the second partial derivative with respect to $x_i$ and $x_j$.
The gradient of $\rho$ is given by 
\begin{align}
  \grad\rho = e^{-v}(-\grad v) = -\rho\,\grad v.
  \label{eq:step-a}
\end{align}
For the second-order part $\sum_{ij}a_{ij}\partial_{ij}\rho$, we have
\begin{align}
  \sum_{ij}a_{ij}\partial_{ij}\rho
  = -\sum_{ij}a_{ij}\partial_j\rho\,\partial_i v
    - \rho\sum_{ij}a_{ij}\partial_{ij}v.
\end{align}
Substituting $\partial_j\rho = -\rho\,\partial_j v$, we have
\begin{align}
  \sum_{ij}a_{ij}\partial_{ij}\rho
  = \rho\sum_{ij}a_{ij}\partial_i v\,\partial_j v
    - \rho\sum_{ij}a_{ij}\partial_{ij}v
  = \rho\bigl[(\grad v)^\top a\,\grad v - \Delta_a v\bigr],
  \label{eq:step-b}
\end{align}
where $\Delta_a v := \sum_{ij}a_{ij}\partial_{ij}v$.
Hence, $L\rho$ can be written as
\begin{align}
\begin{aligned}
L\rho=& b^\top (-\rho\,\grad v)
+ \tfrac{1}{2}\rho\bigl[(\grad v)^\top a\,\grad v - \Delta_a v\bigr]\\
=& \rho\left[-b^\top\grad v - \tfrac{1}{2}\Delta_a v
+ \tfrac{1}{2}(\grad v)^\top a\,\grad v\right] \\
=& \rho\left[-Lv + \tfrac{1}{2}(\grad v)^\top a\,\grad v\right],
\end{aligned}
\label{eq:Lrho-final}
\end{align}
where the last line follows from $b^\top\grad v + \frac{1}{2}\Delta_a v = Lv$.
Using $L\rho - f\rho=0$, and since $\rho>0$, divide by $\rho$, we have
\begin{align}
  -Lv + \tfrac{1}{2}(\grad v)^\top a\,\grad v = f.
\end{align}
Rearranging the above equation, we have 
\begin{align}
    Lv(x) + f(x) = \tfrac{1}{2}(\grad v(x))^\top a(x)\,\grad v(x)
    \quad\text{in }D,
    \qquad v\big|_{\partial D} = \Phi.
  \label{eq:hjb-stationary}
\end{align}

We next show that this coincides with the HJB equation for the OSC \eqref{eq:soc} up to an exit time. Let $\tilde X_t$ solve the controlled SDE \eqref{eq:general-controlled-sde}.  Since $v\in C^2(D)$, It\^o's formula gives
\begin{align}
  dv(\tilde X_t) = \underbrace{\bigl[Lv(\tilde X_t)
    + (\grad v(\tilde X_t))^\top a(\tilde X_t)\,u(\tilde X_t)\bigr]}_{\text{drift}}\,dt
  + \grad v(\tilde X_t)^\top \sigma(\tilde X_t)\,d\tilde W_t.
  \label{eq:ito-general-u}
\end{align}
For $u \in \Ucal$, substituting \eqref{eq:hjb-stationary} into the drift, and adding and subtracting $\frac{1}{2}u^\top a\,u$, we have
\begin{align}
  \text{drift}
  = -\Bigl(f + \tfrac{1}{2}u^\top a\,u\Bigr)
  + \tfrac{1}{2}(\grad v + u)^\top a\,(\grad v + u),
  \label{eq:complete-square-exit}
\end{align}
where the last term $\geq 0$ since $a \succeq 0$.
Equation \eqref{eq:complete-square-exit} also shows that, for each $x\in D$, the mapping $u\mapsto u^\top a\nabla v + \tfrac12 u^\top a u$ is minimized uniquely at $u^*=-\nabla v$, since $a\succ 0$ by Assumption~\ref{A1.elliptic}.
Therefore, \eqref{eq:hjb} is the HJB equation of the OSC \eqref{eq:soc} up to an exit time, with a candidate optimal controller $u^*$ in \eqref{eq:u-star-exit}.

\subsection{Proof of Proposition \ref{prop:gibbs}}
Since $H \geq 0$, we have $H e^{-H}\geq 0$ and hence $-\E_{x}[H e^{-H}]<\infty$. Together with $Z\in(0,\infty)$, hypotheses of \citet[Proposition 2.1]{mitter2003variational} hold on $(\Omega,\Fcal_{\tau_D})$ with reference measure $P^x|_{\Fcal_{\tau_D}}$ and energy $H$. Hence, the claim follows.
\subsection{Proof of Lemma \ref{lemma:girsanov-density}}
\label{pf:lemma:girsanov-density}
Since $\rho\in C^{2,\alpha}(\bar D)$ as established in Lemma \ref{lemma.FK}, and  $-\log$ is smooth on $[m,M]$, the chain rule gives $v=-\log\circ\,\rho\in C^{2,\alpha}(\bar D)$, thus we have $u^* = -\grad v \in C^{1,\alpha}(\bar D)$, hence $u^* \in \Ucal$. By the paragraph preceding Corollary \ref{cor:finite-exit-time-control}, $\tilde b:=b-a\grad v$ is bounded and Lipschitz on $\Rset^n$. Thus, \eqref{eq:general-controlled-sde} has a pathwise unique strong solution for all $t\ge0$ by the standard SDE existence-uniqueness theorem for bounded, Lipschitz coefficients; see, e.g., \cite[Thm.~5.2.1]{oksendal2013stochastic}.

Define an exponential process as
\begin{align}
  Z_t := \exp \left(\int_0^t \left((\sigma^\top u^*)(X_s)\right)^\top dW_s
    - \tfrac12\int_0^t |\sigma^\top u^*(X_s)|^2\,ds\right),
  \qquad t < \tau_D.
  \label{eq:Zt-def}
\end{align}
Since $v\in C^{2,\alpha}(\bar D)$ and $\bar D$ is compact by Assumption \ref{A1}, $\nabla v$ is continuous and hence bounded on $\bar D$. Together with $\sigma$ bounded by Assumption \ref{A1.elliptic}, let $C:=\vnorm{\sigma}_\infty \vnorm{\nabla v}_\infty$ and we have $\left|\sigma^\top u^*\right| \leq C$ on $\bar D$. 
Since $\tau_D<\infty$ a.s.\  by Lemma~\ref{cor:finite-exit-time}, we have $\int_0^{\tau_D} |\sigma^\top u^*(X_s)|^2\,ds \leq C^2 \tau_D < \infty$ a.s.\ Thus the stochastic integral in \eqref{eq:Zt-def} is well defined at $t=\tau_D$, and $t\mapsto Z_t$ is continuous on $[0,\tau_D]$.

\emph{Step 1:}
We compute $Z_t$ as follows. Since $a=\sigma\sigma^\top$ and $u^*=-\grad v$, we have $(\grad v)^\top a\,\grad v = |\sigma^\top u^*|^2$, so the HJB equation \eqref{eq:hjb-stationary} is $Lv = \tfrac12|\sigma^\top u^*|^2 - f$ in $D$.
Applying Lemma~\ref{lem:localized-dynkin} with ${h} = v$ and ${c} = 0$, and taking $k$ large enough that $\sigma_k>t$, which is possible since $\sigma_k\to\tau_D$, for $t<\tau_D$,  we have, 
\begin{align}
  dv(X_t) = Lv(X_t)\,dt + \grad v(X_t)^\top \sigma(X_t)\,dW_t
  = \Bigl(\tfrac12|\sigma^\top u^*(X_t)|^2 - f(X_t)\Bigr)dt
    - \left((\sigma^\top u^*)(X_t) \right)^\top dW_t,
\end{align}
where we used $\grad v=-u^*$ and the HJB equation \eqref{eq:hjb}. 

Rearranging and integrating from $0$ to $t<\tau_D$, we have
\begin{align}
  \int_0^t \left((\sigma^\top u^*)(X_s) \right)^\top dW_s
  - \tfrac12\int_0^t |\sigma^\top u^*(X_s)|^2\,ds
  = v(x) - v(X_t) - \int_0^t f(X_s)\,ds.
\end{align}
Taking exponential and using $e^{-v}=\rho$, we have
\begin{align}
  Z_t = \frac{\rho(X_t)}{\rho(x)}
  \exp \left(-\int_0^t f(X_s)\,ds\right), \qquad t<\tau_D.
  \label{eq:Zt-explicit}
\end{align}

\emph{Step 2:}
We now take $t \to \tau_D$ in \eqref{eq:Zt-explicit}. By Lemma~\ref{cor:finite-exit-time}, we have $\tau_D<\infty$, $P^x$-a.s., thus $X_t\to X_{\tau_D}\in\partial D$ by continuity of the
path, and $\rho(X_t)\to\rho(X_{\tau_D})=e^{-\Phi(X_{\tau_D})}$ by continuity of $\rho$ on $\bar D$ together with the boundary condition $\rho|_{\partial D}=e^{-\Phi}$ in \eqref{eq:elliptic-bvp}. Moreover, $\int_0^tf(X_s)\,ds\to\int_0^{\tau_D}f(X_s)\,ds$ by continuity of $t\mapsto\int_0^tf(X_s)\,ds$, the limit is finite since $f$ is bounded by Assumption \ref{A3} and $\tau_D<\infty$ a.s. Hence, we have
\begin{align}
  Z_{\tau_D}
  = \frac{e^{-\Phi(X_{\tau_D})}}{\rho(x)}\exp\Bigl(-\int_0^{\tau_D}f(X_s)\,ds\Bigr)
  = \frac{e^{-H(\chi)}}{\rho(x)},
  \qquad P^x\text{-a.s.}
  \label{eq:Z-at-tau}
\end{align}
Taking expectations, the Feynman--Kac identity $\rho(x)=\E_{x}[e^{-H(\chi)}]$ of \eqref{eq:rho-exit} gives
\begin{align}
  \E_{x}\bigl[Z_{\tau_D}\bigr]=\frac{\rho(x)}{\rho(x)}=1 .
\end{align}

\emph{Step 3:}
Since $\bar D$ is compact and $\rho\in C(\bar D)$, we have $0<m\le\rho\le M<\infty$ on $\bar D$; and $\exp(-\int_0^tf\,ds)\in(0,1]$ since $f\ge0$. Hence \eqref{eq:Zt-explicit} gives the uniform, deterministic bound
\begin{align}
  0 < Z_{t\wedge\tau_D} \le \frac{M}{\rho(x)}
  \qquad\text{for every }t\ge0\text{ and a.e. }\omega .
  \label{eq:Z-bounded}
\end{align}
The stopped process $Z_{\cdot\wedge\tau_D}$ is a stochastic exponential, hence a local martingale; by \eqref{eq:Z-bounded} it is bounded, and a bounded local martingale is a true martingale. Therefore
$\E_{x}[Z_{t\wedge\tau_D}]=1$ for every $t\ge0$.

\emph{Step 4:}
Finally, define $Q$ on $\Fcal_{\tau_D}$ via $dQ/dP^x=Z_{\tau_D}$, which is a probability measure by the above. By Girsanov's theorem, we have $W_t-\int_0^t\sigma^\top u^*(X_s)\,ds$ is a $Q$-Brownian motion on $[0,\tau_D]$. Thus under $Q$, the process solves the controlled SDE \eqref{eq:general-controlled-sde} with $u=u^*$ up to $\tau_D$. By uniqueness in law for said SDE, established above, we have $Q=\tilde P^{*,x}$ on $\Fcal_{\tau_D}$, which together with \eqref{eq:Z-at-tau} proves \eqref{eq:girsanov-density}.

\subsection{Proof of Theorem~\ref{thm:exit-time-SOC}}
\label{pf:thm:exit-time-SOC}
\medskip\noindent\emph{Part (i).}
We first consider $u\in\Ucal$, and consider the optimal control $u^*$ at the end as the special case. 
We proceed in four steps: (1) applying It\^o's formula for $v(\tilde X_t)$ under $u \in \Ucal$; (2) applying Lemma \ref{lem:localized-dynkin} to integrate $v$ on $\bar D$; (3) passing $k\to\infty$; and (4) rearranging and specializing to $u=u^*$.

\medskip\noindent\emph{Step 1:}
From It\^o's formula and \eqref{eq:complete-square-exit}, for the controlled process $\tilde X$,we have
\begin{align}
  dv(\tilde X_t)
  = -\Bigl(f(\tilde X_t)+\tfrac12 u^\top a\,u(\tilde X_t)\Bigr)dt
  + \tfrac12\bigl(\grad v+u\bigr)^\top a\,\bigl(\grad v+u\bigr)(\tilde X_t)\,dt
  + \grad v(\tilde X_t)^\top \sigma(\tilde X_t)\,d\tilde W_t .
  \label{eq:ito-general-parti}
\end{align}

\medskip\noindent\emph{Step 2:}
The above equation is derived from $v \in C^2(D)$, so it may be integrated only while the trajectory remains in $D$. We therefore apply Lemma~\ref{lem:localized-dynkin} with drift $b + au$, ${h} = v$ and ${c}\equiv 0$ along the controlled diffusion process \eqref{eq:general-controlled-sde}, where admissibility of $u$ and Assumption~\ref{A1.elliptic} make the drift term bounded and Lipschitz. Substituting the drift decomposition in \eqref{eq:complete-square-exit} into \eqref{eq:localized-dynkin}, we have
\begin{align}
  \underbrace{{\tE}\bigl[v(\tilde X_{\sigma_k})\bigr]}_{:=T_1} - v(x)
  = -\underbrace{{\tE} \left[\int_0^{\sigma_k}\Bigl(f+\tfrac12u^\top a\,u\Bigr)dt\right]}_{:=T_2}
  + \underbrace{{\tE} \left[\int_0^{\sigma_k}\tfrac12(\grad v+u)^\top a\,(\grad v+u)\,dt\right]}_{:=T_3}.
  \label{eq:localized-identity-parti}
\end{align}

\medskip\noindent\emph{Step 3:}
We next study the three terms $T_1$-$T_3$ in \eqref{eq:localized-identity-parti} separately.  For the boundary term $T_1$, recall by Corollary~\ref{cor:finite-exit-time-control}, $\tilde\tau_D<\infty$ a.s., and we have shown before that $\sigma_k \to \tilde\tau_D$ as $k\to\infty$ in Lemma \ref{lem:localized-dynkin}. 
Hence $\tilde X_{\sigma_k}\to\tilde X_{\tilde\tau_D}$ a.s.\ by continuity of the path. Since the controlled process exits $D$ at time $\tilde\tau_D$, we have $\tilde X_{\tilde\tau_D}\in\partial D$ a.s., so continuity of $v$ on $\bar D$ together with the boundary condition $v|_{\partial D}=\Phi$ gives $v(\tilde X_{\sigma_k}) \;\to\; v(\tilde X_{\tilde\tau_D}) = \Phi(\tilde X_{\tilde\tau_D})$ a.s.
Moreover $v\in C(\bar D)$ with $\bar D$ compact gives $|v|\le M:=\max_{\bar D}|v|<\infty$, so $|v(\tilde X_{\sigma_k})|\le M$ for every $k$ and a.e.\ $\omega$. By bounded convergence theorem, we thus have
\begin{align}
  \lim_{k\to\infty}{\tE}\bigl[v(\tilde X_{\sigma_k})\bigr]
  = {\tE}\bigl[\Phi(\tilde X_{\tilde\tau_D})\bigr].
  \label{eq:limit-boundary-parti}
\end{align}

For the running cost $T_2$, since $f\ge0$ by Assumption~\ref{A3} and $a\succ0$ by Assumption~\ref{A1.elliptic}, we have that $f+\tfrac12u^\top a\,u$ is nonnegative. As $\sigma_k$ is nondecreasing with $\sigma_k\to\tilde\tau_D$ from Lemma \ref{lem:localized-dynkin}, the integrals $\int_0^{\sigma_k}(f+\tfrac12u^\top a\,u)\,dt$ are nondecreasing in $k$ and converge pointwise a.s.\ to $\int_0^{\tilde\tau_D}(f+\tfrac12u^\top a\,u)\,dt$. By monotone convergence theorem, we have
\begin{align}
\lim_{k\to\infty}{\tE} \left[\int_0^{\sigma_k}\Bigl(f+\tfrac12u^\top a\,u\Bigr)dt\right]
={\tE} \left[\int_0^{\tilde\tau_D}\Bigl(f+\tfrac12u^\top a\,u\Bigr)dt\right].
\label{eq:limit-running-parti}
\end{align}

For the last term $T_3$, proceed similarly as above, since $a\succ0$, $\tfrac12(\grad v+u)^\top a\,(\grad v+u)$ is nonnegative, so monotone convergence theorem gives
\begin{align}
\lim_{k\to\infty}{\tE} \left[\int_0^{\sigma_k}\tfrac12(\grad v+u)^\top a\,(\grad v+u)\,dt\right]
  = {\tE} \left[\int_0^{\tilde\tau_D}\tfrac12(\grad v+u)^\top a\,(\grad v+u)\,dt\right].
  \label{eq:limit-remainder-parti}
\end{align}

We next show that all three limits are finite. Indeed, by Assumption~\ref{A3}, $\Phi=-\log g$ is continuous on the compact set $\partial D$, hence bounded. Let $C:= \sup_{\bar D}\Bigl(f+\tfrac12u^\top a\,u\Bigr)$ and $C':= \sup_{\bar D}\tfrac12(\grad v+u)^\top a\,(\grad v+u)$. Both terms are finite because $f$ is bounded by Assumption~\ref{A3}, $a$ is bounded by Assumption~\ref{A1.elliptic}, $u \in \Ucal$ is bounded on the compact $\bar D$ by definition, and $\grad v$ is bounded on $\bar D$ since $v\in C^{2,\alpha}(\bar D)$. The two integrals in \eqref{eq:limit-running-parti} and  \eqref{eq:limit-remainder-parti} are dominated by $C \tilde\tau_D$ and $C' \tilde\tau_D$, respectively, whose expectations are finite since ${\tE}\bigl[\tilde\tau_D\bigr]<\;\infty$ by Corollary~\ref{cor:finite-exit-time-control}.
Combining \eqref{eq:limit-boundary-parti}--\eqref{eq:limit-remainder-parti} into \eqref{eq:localized-identity-parti}, we have
\begin{align}
  {\tE}\bigl[\Phi(\tilde X_{\tilde\tau_D})\bigr] - v(x)
  = -{\tE} \left[\int_0^{\tilde\tau_D}\Bigl(f+\tfrac12u^\top a\,u\Bigr)dt\right]
  + {\tE} \left[\int_0^{\tilde\tau_D}\tfrac12(\grad v+u)^\top a\,(\grad v+u)\,dt\right],
  \label{eq:limit-identity-parti}
\end{align}
where every term is finite.

\emph{Step 4:} Rearranging \eqref{eq:limit-identity-parti}, and using $\grad v+u=u-u^*$ as $u^*=-\grad v$, we obtain
\begin{align}
\begin{aligned}
  J(u,x)
  &= {\tE}\bigl[\Phi(\tilde X_{\tilde\tau_D})\bigr]
   + {\tE}\left[\int_0^{\tilde\tau_D}
   \Bigl(f+\tfrac12u^\top a u\Bigr)dt\right]\\
  &= v(x) 
  + {\tE}\left[\int_0^{\tilde\tau_D}\tfrac12(u-u^*)^\top a(u-u^*) dt\right] \\
  &\geq  v(x),
\end{aligned}
\label{eq:gap-exit}
\end{align}
where the inequality follows from $a\succ0$ by Assumption~\ref{A1.elliptic}. 
Since $u\in\Ucal$ was arbitrary, $v(x)$ is a lower bound. Here, equality holds if and only if ${\tE}\bigl[\int_0^{\tilde\tau_D}\tfrac12(u-u^*)^\top a\,(u-u^*)\,dt\bigr]=0$. Since the integrand is nonnegative, this holds if and only if $\int_0^{\tilde\tau_D}\tfrac12(u-u^*)^\top a\,(u-u^*)(\tilde X_t)\,dt=0$, $\tP^{u,x}$-a.s. Again, by nonnegativity, it holds if and only if $(u-u^*)^\top a\,(u-u^*)(\tilde X_t)=0$ for Lebesgue-almost every $t\in[0,\tilde\tau_D(\omega)]$,$\tP^{u,x}$-a.s. By uniform ellipticity in Assumption~\ref{A1.elliptic}, we have $(u-u^*)^\top a\,(u-u^*)\ge c\,|u-u^*|^2$ with $c>0$. So this is equivalent to $u(\tilde X_t)=u^*(\tilde X_t)$ for Lebesgue-almost every $t\in[0,\tilde\tau_D]$, $\tP^{u,x}$-a.s.

It remains to check that this bound is attained. Since
$v\in C^{2,\alpha}(\bar D)$, we have $u^*=-\grad v\in C^{1,\alpha}(\bar D)$, hence $u^*\in\Ucal$ and \eqref{eq:gap-exit} applies. The remainder term then vanishes identically, and we have $J(u^*,x)=v(x)$. Hence, the infimum is attained, and for every $x\in D$, we have
\begin{align}
  v(x)=J(u^*,x)=\min_{u\in\Ucal}J(u,x).
\end{align}
This completes the proof of Part (i).

\medskip\noindent\emph{Part (ii)}
Since $\tau_D<\infty$, $P^x$-a.s. as in Lemma~\ref{cor:finite-exit-time}, $H(\chi)$ is well-defined and finite $P^x$-a.s., and by the Feynman--Kac representation of $\rho$, we can define
\begin{align}
  Z := \E_{x}\bigl[e^{-H(\chi)}\bigr] = \rho(x) \in (0,1],
\end{align}
so $Z\in(0,\infty)$ as required by \cite[Proposition 1]{raginsky2026}. 
Moreover $f\ge0$ by Assumption~\ref{A3} and $\Phi=-\log g\ge0$ since $g\le1$ by the same assumption, so $H(\chi)\ge0$ $P^x$-a.s. Hence the integrand
$-H(\chi)e^{-H(\chi)}$ is non-positive and $-\E_{x}\bigl[H(\chi)\,e^{-H(\chi)}\bigr]  <\infty$. So both assumptions of Proposition \ref{prop:gibbs} are satisfied, and applying it on $(\Omega,\Fcal_{\tau_D})$ with $P=P^x$ and $H=H(\chi)$, for every admissible path measure $\tP\ll P^x$, we have
\begin{align}
  \DKL(\tP\|P^x) + \E[H(\chi)] \;\geq\; -\log\rho(x) = v(x),
  \label{eq:gibbs-applied}
\end{align}
with equality if and only if $d\tP/dP^x = e^{ -H(\chi)}/\rho(x)$.
By Part~(i), we have $v(x)=J(u^*,x)$, so \eqref{eq:gibbs-applied} is \eqref{eq:variational-ineq-exit}. 
By Lemma~\ref{lemma:girsanov-density}, $\tP^{*,x}$ of the optimally-controlled process satisfies 
\begin{align}
  \left.\frac{d\tP^{*,x}}{dP^x}\right|_{\Fcal_{\tau_D}}
  = \frac{e^{-H(\chi)}}{\rho(x)}
\end{align}
Hence equality in \eqref{eq:gibbs-applied}, and so in \eqref{eq:variational-ineq-exit}, holds if and only if $\tP=\tP^{*,x}$.


\section{Proof of Identity \eqref{eq:committor-bvp} in Section \ref{sec:committor}}
\label{pf:sec:committor}

Let Assumption \ref{A1.elliptic} hold. By construction in Section \ref{sec:committor}, $A_i,A_j$ are closed, and $\partial A_i\cap\partial A_j=\emptyset$. Assume that $D_q=\Xset\setminus(A_i\cup A_j)$ satisfies Assumption \ref{A2}, with $\partial D_q=\partial A_i\cup\partial A_j$.
Denote $\varphi:\partial D_q\to\{0,1\}$ with $\varphi=0$ on $\partial A_i$ and $\varphi=1$ on $\partial A_j$. Since $\partial A_i$ and $\partial A_j$ are disjoint compact sets, $\varphi$ is locally constant, hence $\varphi\in C^{2,\alpha}(\partial D_q)$. By \cite[Ch.~III, Theorem~1.1 and Theorem~1.2]{ladyzhenskaya1968linear}, the boundary value problem
\begin{align}
L w=0 \text{ in } D_q,\qquad w|_{\partial D_q}=\varphi,
\label{eq:w-committor}
\end{align}
has a unique solution $w\in C^{2,\alpha}(\bar D_q)$. In particular, $w$ is continuous and bounded on the compact set $\bar D_q$.

We next show $q = w$. For $x\in D_q$, let $\tau:=\inf\{t>0: X_t\notin D_q\}$ be the first time the trajectory leaves $D_q$. 
Since $\Xset \setminus D_q=A_i\cup A_j$, we have $\tau=\tau_{A_i}\wedge\tau_{A_j}$.
Let $D_k$, $\sigma_k$ be as in Lemma \ref{lem:localized-dynkin} with $D=D_q$. Apply Lemma~\ref{lem:localized-dynkin}(ii) with $h=w$ and $c = 0$, we have $Z_t = 1$. Since $L w=0$ in $D_q$, from \eqref{eq:localized-dynkin}, we have
\begin{align}
\E_x\big[w(X_{\sigma_k})\big]=w(x), \qquad k\in \Nset.
\label{eq:dynkin-committor}
\end{align}
Here $w(X_{\sigma_k})$ is well defined since $X_{\sigma_k}\in\bar D_q$ by the case split in the proof of Lemma~\ref{lem:localized-dynkin}.

We next take $k \to \infty$. From Lemma~\ref{lem:localized-dynkin}(i), we have $\sigma_k \to \tau$. Since $\tau<\infty$ a.s., $D_q$ is open, and by path continuity, we have $X_{\sigma_k}\to X_\tau \in\partial D_q$ a.s. Since $w$ is continuous on $\bar D_q$, we have $w(X_{\sigma_k})\to w(X_\tau)=\varphi(X_\tau)$ a.s. 
Furthermore, $|w(X_{\sigma_k})|\leq \max_{\bar D_q}|w|<\infty$ for $k \in \Nset$. By the bounded convergence theorem, we then have 
\begin{align}
w(x)=\E_x\big[\varphi(X_\tau)\big]=P^x\big(X_\tau \in \partial A_j\big),
\end{align}
where the second equality follows since $\varphi(X_\tau) = \bm{1}_{\partial A_j} (X_\tau)$.

Finally, we show $P^x\big(X_\tau\in \partial A_j\big) = q(x)$. Since $A_i,A_j$ are closed by construction, and paths are continuous, and by definition of $\tau$, we have $X_{\tau_{A_k}}\in A_k$ whenever $\tau_{A_k}<\infty$ for $k\in\{i,j\}$. If $X_\tau\in\partial A_j\subset A_j$, then $\tau\neq\tau_{A_i}$, since otherwise $X_\tau\in A_i\cap A_j=\emptyset$, and we have $\tau=\tau_{A_j}<\tau_{A_i}$. Conversely, if $\tau_{A_j}<\tau_{A_i}$, then $\tau=\tau_{A_j}$, and we have $X_\tau\in\partial A_j$. Thus  $\{X_\tau \in \partial A_j\}=\{\tau_{A_j}<\tau_{A_i}\}$, and for $x\in D_q$, we have
\begin{align}
w(x)=P^x\big(\tau_{A_j}<\tau_{A_i}\big)=q(x). 
\end{align}
By definition of $q$, for $x\in\partial A_i$ ($x\in\partial A_j$),  we have $q(x)=0=w(x)$ ($q(x)=1=w(x)$). Hence $q=w$ on $\bar D_q$. Therefore, $q\in C^{2,\alpha}(\bar D_q)$, and $q$ is the unique solution of \eqref{eq:committor-bvp}.

\section{Detailed Computation of the Committor Function from Section \ref{sec:committor}}
\label{supp.committor}

We first construct the spectral term $\tilde{\psi}$ that separates the two cores using the leading eigenfunctions of $\Lmd$. Consider $A_k$ and the set of leading eigenfunctions $\psi_l$, for $l=1,\dots,m-1$. With the core level $\theta^{(k)}_l$ defined in \eqref{eq:core-level}, we define $\hat\psi(x):= \sum_{l=1}^{m-1} c_l\,\psi_l(x)$, where $c_l := \theta^{(j)}_l- \theta^{(i)}_l$ measures how much the $l$-th mode varies between $A_i,A_j$. Define the core level of $\hat \psi$ as $\bar \theta^{(k)}:= \frac{1}{|\Ical_k|}\sum_{i\in\Ical_k}\hat \psi(x_{i})$. Substituting $\hat\psi$, we have $\bar \theta^{(k)} = \sum_{l} c_l \theta_l^{(k)}$. Let $c = [c_1,\cdots, c_{m-1}]$. We recenter and rescale $\hat{\psi}$ to get ${\tilde{\psi}}$ as
\begin{align}
    {\tilde{\psi}}(x) := \frac{{\hat{\psi}}(x)- \bar \theta^{(i)}}{\|c\|^2},
    \qquad \bar \theta^{(i)} = \sum_{l} c_l \theta_l^{(i)}.
    \label{eq:tildepsi}
\end{align}
With this choice of $c$, we have $\bar \theta^{(j)} - \bar \theta^{(i)} = \|c\|^2$, and $\tilde \psi$ maps the centroid of $A_i$ to $0$ and that of $A_j$ to $1$. Since each $\psi_l$ is an eigenfunction, we have $\nabla {\hat{\psi}} = \sum_{l=1}^{m-1} c_l\,\nabla\psi_l$, and ${\Lmd} {\hat{\psi}} = \sum_{l=1}^{m-1} c_l\,\lambda_l\,\psi_l$.
Applying $\Lmd$ to~\eqref{eq:tildepsi} and using that $\bar\theta^{(i)}$ and $\|c\|^2$ are constants, so that $L\bar\theta^{(i)}=0$, we have $L\tilde\psi =L\hat\psi/ \|c\|^2
  =\frac{1}{\|c\|^2}\sum_{l=1}^{m-1} c_l\,\lambda_l\,\psi_l$.
Using linearity of ${L}$, we obtain ${L}h={L}q-{L}\tilde\psi=-{L}\tilde\psi$ in $D_q$. Together with the boundary conditions $q|_{\partial A_i}=0$ and $q|_{\partial A_j}=1$, we have that $h$ solves
\begin{align}
    {L}h = -\,{L}\tilde\psi \ \text{ in } D_q,
    \qquad
    h|_{\partial A_i} = -\tilde\psi|_{\partial A_i},
    \qquad
    h|_{\partial A_j} = 1-\tilde\psi|_{\partial A_j}.
    \label{eq:interior-in-h}
\end{align}

We next approximate the remainder $h$ in the ansatz $q = \tilde{\psi} + h$ in $\Hcal_\kappa$. 
To solve \eqref{eq:interior-in-h} with $\tilde{\psi}$ in \eqref{eq:tildepsi}, let $\{x_l\}_{l\in\Ical_q}\subset D_q$ with $\Ical_q=\{1,\dots,N_q\}$, and consider solving for a $h\in \Hcal_\kappa$ of the form
\begin{align}
    h(x)=\sum_{l\in\Ical_q} \alpha_l\,\kappa(x,x_{l}),
    \qquad \alpha=(\alpha_l)_{l\in\Ical_q}.
    \label{eq:h-expansion}
\end{align}
The problem then reduces to solving for $\alpha$. To this end, we impose ~\eqref{eq:interior-in-h} it at interior points $\{x_{n}\}_{n\in \Jcal_q}\subset D_q$ where the residual is minimized, and boundary points $\{x_{n}\}_{n \in \Jcal_\partial} \subset \partial A_i\cup\partial A_j$ where the boundary condition is imposed. 
Specifically, define the matrix $A \in \Rset^{|\Jcal_q| \times |\Ical_q|}$ and vector $g \in \Rset^{|\Jcal_q|}$ with entries
\begin{align}
   A_{nl} := ({L}\kappa)(x_{n},x_{l}),
    \qquad
    [g]_n := ({L}\tilde\psi)(x_{n})
      = \frac{1}{\|c\|^2}
      \sum_{l=1}^{m-1} c_l\lambda_l\psi_l(x_{n}),
    \label{eq:A-matrix}
\end{align}
for $n\in \Jcal_q$. Similarly, define $K_\partial \in \Rset^{|\Jcal_\partial| \times |\Ical_q|}$ and $b \in \Rset^{|\Jcal_\partial|}$ as
\begin{align}
    (K_\partial)_{nl} := \kappa(x_{n},x_{l}),
    \qquad
    (b)_n := 
    \begin{cases}
      -\tilde\psi(x_{n}), & x_{n}\in\partial A_i,\\
      1-\tilde\psi(x_{n}), & x_{n}\in\partial A_j,
    \end{cases}
    \label{eq:bdry-block}
\end{align}
for $n\in \Jcal_\partial$.
Here, ${L}\kappa$ can be computed using the derivative of the kernel function $\kappa$. Specifically, for the Gaussian kernel $\kappa(x,y)=\exp\big(-\|x-y\|^2/2\sigma^2\big)$ with bandwidth $\sigma$, writing $z:=x-y$, we have
\begin{align}
    \nabla_x\kappa(x,y) = -\frac{z}{\sigma^{2}}\,\kappa(x,y),
    \qquad
    \partial_{x_ix_j}\kappa(x,y)
      = \Big(\frac{z_iz_j}{\sigma^{4}}-\frac{\delta_{ij}}{\sigma^{2}}\Big)\kappa(x,y).
    \label{eq:gaussian-derivs}
\end{align}
Hence, we have
\begin{align}
    ({L}\kappa)(x,y)
    = \left(
        -\frac{b(x)\cdot z}{\sigma^{2}}
        + \frac{z^\top a(x)\,z}{2\sigma^{4}}
        - \frac{\operatorname{tr}a(x)}{2\sigma^{2}}
      \right)\kappa(x,y).
\end{align}

We next solve for $\alpha$ in \eqref{eq:h-expansion} by imposing the boundary condition at $\{x_n\}_{n\in\Jcal_\partial}$ and minimizing the least-squares sense with Tikhonov regularization as
\begin{align}
\begin{aligned}
\min_{\alpha}\ \
& \big\|A\alpha + g\big\|^2 
+ \lambda_q\,\alpha^\top K \alpha,\\
\text{s.t.}\ \ & K_\partial\,\alpha = b ,
\end{aligned}
\label{eq:committor-qp}
\end{align}
for $\lambda_q\ge0$ and $K_{ll'}=\kappa(x_{l},x^{l'})$. 
For such an equality-constrained quadratic program, define a multiplier $\mu \in \Rset^{|\Jcal_\partial|}$. The Lagrangian is then $\|A\alpha+g\|^2+\lambda_q\alpha^\top K\alpha+2\mu^\top(K_\partial\alpha-b)$, and the KKT conditions in $(\alpha,\mu)$ are
\begin{align}
    \begin{bmatrix}
      A^\top A 
      + \lambda_q K & K_\partial^\top \\[2pt]
      K_\partial & 0
    \end{bmatrix}
    \begin{bmatrix}\alpha\\ \mu\end{bmatrix}
    =
    \begin{bmatrix}-A^\top g\\ b\end{bmatrix}.
    \label{eq:committor-kkt}
\end{align}
Solving for $\hat{\alpha}$ via the linear system \eqref{eq:committor-kkt} and plugging this into $q = \tilde{\psi} + h$, the committor function can be approximated via \eqref{eq:committor-final}.

\section{Details Behind the Computation of $\rho$ from Section \ref{sec:koopman-exit}}
\label{sec:supp.rho}

We start by constructing $\rho_0$. Consider $x \in D$, define $\tau_i := \inf\{t>0 : X_t \in A_i\}$, and $\tau_j := \inf\{t>0 : X_t \in A_j\} = \tau_D$. Let $R := \{\tau_j < \tau_i\}$ be the reactive event, i.e., the trajectory reaches the target set before returning to the source set. By definition of the committor function in Section~\ref{sec:committor}, we have $q^{(i,j)}(x) = {P^x}(R)$ where ${P^x}$ denotes the probability of starting from $x$. Conditioning $\E_x \left[e^{-\beta {H_0}}\right]$ in \eqref{eq:W-def} on $R$ and $R^c$, and using the law of total expectation, we have
\begin{align}
\begin{aligned}
\rho(x) 
=& \E_x \left[e^{-\beta {H_0}}\right]\\
=& {P^x}(R) \underbrace{\E_x \left[e^{-\beta {H_0}} \,\middle|\, R\right]}_{:=\rho_{\mathrm{R}}(x)}
+ {P^x}(R^c) \underbrace{\E_x \left[e^{-\beta {H_0}} \,\middle|\, R^c\right]}_{:=\rho_{\mathrm{NR}}(x)} \\
=& q^{(i,j)}(x) \rho_{\mathrm{R}}(x) 
+ \bigl(1-q^{(i,j)}(x)\bigr)\,\rho_{\mathrm{NR}}(x),
\end{aligned}
\label{eq:committor-is-splitting}
\end{align}
where $\rho_{\mathrm{R}}, \rho_{\mathrm{NR}} \in (0,1]$ since ${H_0} \geq 0$, with the convention that a term is omitted when its conditioning event has probability zero. Thus $\rho(x)$ is a convex combination of $\rho_{\mathrm R}(x)$ and $\rho_{\mathrm{NR}}(x)$ with weight being $q(x)$.
In practice, we use $\hat q$ from \eqref{eq:committor-final} as an approximation of $q$.

On $R^c$, the trajectory reaches $A_i$ at time $\tau_i < \tau_j$. We therefore split ${H_0} = {H_0^{i}} + {H_0^{ij}}$ with ${H_0^{i}} := \int_0^{\tau_i}(1-q^{(i,j)})\,dt$ and ${H_0^{ij}} := \int_{\tau_i}^{\tau_j}(1-q^{(i,j)})\,dt$. Conditioning on the path up to time $\tau_i$, by the Markov property, we have $\E\bigl[e^{-\beta {H_0^{ij}}} \mid \mathcal F_{\tau_i}\bigr] = \rho(X_{\tau_i})$ where $X_{\tau_i} \in \partial A_i$. Since $e^{-\beta {H_0}} =  e^{-\beta {H_0^{i}}} e^{-\beta {H_0^{ij}}}$, we have
\begin{align}
  \rho_{\mathrm{NR}}(x) = \E_x \left[e^{-\beta {H_0^{i}}}\,\rho(X_{\tau_i}) \middle|  R^c\right] .
\end{align}
For $\rho(X_{\tau_i})$, notice that the trajectory lingers inside $A_i$, and hence $\rho(X_{\tau_i}) \approx \rho|_{A_i}$. To compute the latter, let $y \in A_i$. Since $q^{(i,j)} \equiv 0$ on $A_i$ and the trajectory spends nearly all of time $[0,\tau_j]$ inside $A_i$, we have $H_0^{i} \approx \tau_j$. Since the system is metastable, we assume the first hitting time is asymptotically exponential with mean $T_0 = 1/|\lambda_1|$, where $\lambda_1$ is the leading nontrivial eigenvalue of the Koopman generator. In other words,, $\tau_j$ has density $T_0^{-1}e^{-t/T_0}$ on $t \ge 0$, and we have
\begin{align}
  \rho_A  \approx  \E \left[e^{-\beta \tau_j}\right]
   =  \int_0^\infty \frac{1}{T_0}\,e^{-t/T_0}\,e^{-\beta t}\,dt
   =  \frac{1}{1+\beta T_0}  =:  r .
  \label{eq:rho-A}
\end{align}
Hence, we have
\begin{align}
\begin{aligned}
\rho(x) 
=& q^{(i,j)}(x) \rho_{\mathrm{R}}(x)
+ \left[\bigl(1-q^{(i,j)}(x)\bigr) r\right] \,\E_x \left[e^{-\beta {H_0^{i}}} \,\middle|\, R^c\right].
\end{aligned}
\end{align}

On $R$, the trajectory goes from $x$ to $A_j$ without visiting $A_i$, hence $\rho_{\mathrm{R}}=\E_x \left[e^{-\beta {H_0}} \,\middle|\, R\right]$ lies in $(0,1]$ with ${H_0}$ the duration of reactive crossing. For term $\E_x\left[e^{-\beta {H_0^{i}}} \middle|R^c\right]$, on $R^c$, ${H_0^{i}}$ is the time of a non-reactive path falling back into $A_i$. Therefore, neither term encodes the metastability of the system, and the metastability enters only through $q^{(i,j)}(x)$ and $\bigl(1-q^{(i,j)}(x)\bigr) r$. We therefore consider the factorization of $\rho$ as
\begin{align}
\begin{aligned}
\rho =& \left( q^{(i,j)}(x)+ \left(1- q^{(i,j)}(x)\right) r\right) \times \frac{\rho}{q^{(i,j)}(x)+ \bigl(1- q^{(i,j)}(x)\bigr) r}.
\end{aligned}
\label{eq.rho-factor}
\end{align}
Replacing $q^{(i,j)}$ by its approximation $\hat q$ from \eqref{eq:committor-final}, we define $\rho_0(x):= \hat q(x)+ \bigl(1-\hat q(x)\bigr) r$, and the remainder is the ${H_0}$.

We next compute ${w}$ in the ansatz $\hat{\rho}(x) := \rho_0(x)\,w(x)$. Since $\rho_0$ is affine in $\hat q$, we also have
\begin{align}
  \nabla \rho_0 = (1-r)\,\nabla \hat q ,
  \qquad
  \Lmd \rho_0 = (1-r)\,\Lmd \hat q ,
  \label{eq:rho0-derivs}
\end{align}
By the product rule, we have 
\begin{align}
\Lmd(\rho_0 w)=\rho_0 \Lmd w + w \Lmd \rho_0 +\nabla\rho_0^\top a \nabla w.
\end{align}
Thus $\Lmd\rho-f_\beta\rho=0$ in \eqref{eq:elliptic-bvp} can be written in terms of ${w}$ as
\begin{align}
    \rho_0\,\Lmd w
     +  w\,\big(\Lmd\rho_0 - f_\beta\,\rho_0\big)
     +  \nabla\rho_0^\top a\,\nabla w
     = 0
    \qquad\text{in } D.
    \label{eq:bvp-in-w}
\end{align}
To solve \eqref{eq:bvp-in-w}, define two matrices in $\Rset^{|I_D|\times N}$ as
\begin{align}
  (K_D)_{li} = \kappa(x_l,x_i),
  \qquad
  ({\Lmd K})_{li} = (\Lmd\kappa)(x_l,x_i),
  \label{eq:eval-matrices}
\end{align}
for $x_l \in I_D$ and $i = 1,\dots,N$.
In addition, define $\Gamma_D \in \Rset^{|I_D|\times N}$, and $P,Q \in \Rset^{|I_D|\times|I_D|}$ as
\begin{align}
\begin{aligned}
&(\Gamma_D)_{li} = \nabla\rho_0(x_l)^\top a(x_l)\,\nabla_x\kappa(x_l,x_i),\\
&P:=\operatorname{diag}\big(\rho_0(x_l)\big),
\quad
Q:=\operatorname{diag}\big(\Lmd\rho_0(x_l)-f_\beta(x_l)\rho_0(x_l)\big),
\end{aligned}
\label{eq:eval-matrices-2}
\end{align}
where $\nabla_x \kappa$ is available in closed form as in \eqref{eq:gaussian-derivs}, and $\nabla \hat{q}=\nabla\tilde\psi+\sum_{l}\alpha_l\nabla_x\kappa(\cdot,x_l)$ from \eqref{eq:committor-final}.
Using the chain rule and \eqref{eq:rho-kernel-expansion}, we also have
\begin{align}
\Lmd\rho_0 
= \rho_0
\Big(\eta\,\Lmd \hat{q}
+ \tfrac{\eta^{2}}{2}\,\nabla \hat{q}^\top a\,\nabla \hat{q}\Big).
\end{align}
Hence \eqref{eq:eval-matrices} together with \eqref{eq:eval-matrices-2}, the residual of \eqref{eq:bvp-in-w} at the interior points can thus be written as 
\begin{align}
  \big[P\,{\Lmd K} + Q\,K_D + \Gamma_D\big]\,c
   =:  {S_D}\,c,
  \qquad
  {S_D} \in \Rset^{|I_D|\times N}.
  \label{eq:residual}
\end{align}

Likewise, with $b_l:=e^{-\Phi(x_l)}/\rho_0(x_l)$ for $x_l\in I_{\partial D}$ and $(K_{\partial D})_{li} := \kappa(x_l,x_i)$, the boundary condition can be written as $K_{\partial D}\,c = b$.
We then consider the following regularized optimization problem
\begin{align}
  \min_{c\in\Rset^N}\ \ & \|S_D\,c\|^2 + \lambda\, c^\top K c \nonumber\\
  \text{s.t.}\ \ & K_{\partial D}\,c = b ,
  \label{eq:coeff-solve-kernel}
\end{align}
for $\lambda\ge0$, and $K\in\Rset^{N\times N}$ is the Gram matrix, $K_{ii'} = \kappa(x_i,x^{i'})$, so that $c^\top K c = \|w\|^2_{\Hcal_k}$. 
Here, the boundary condition is imposed exactly, while the interior residual is minimized in the least squares sense with Tikhonov regularization. Here, \eqref{eq:coeff-solve-kernel} is again an equality-constrained convex quadratic program. We define the Lagrangian $\|S_Dc\|^2+\lambda c^\top Kc+2\mu^\top(K_{\partial D}c-b)$, and obtain the KKT condition, which is a linear system in $c$ and the Lagrange multiplier $\mu$
\begin{align}
  \begin{pmatrix}
    S_D^\top S_D + \lambda K & K_{\partial D}^\top \\
    K_{\partial D} & 0
  \end{pmatrix}
  \begin{pmatrix} c \\ {\mu} \end{pmatrix}
  = \begin{pmatrix} 0 \\ b \end{pmatrix}.
  \label{eq:kkt-kernel}
\end{align}
With $c$ determined by \eqref{eq:kkt-kernel}, applying the chain rule to $\log\hat{\rho}=\log\rho_0+\log w$ with $\rho_0$ in \eqref{eq:rho-kernel-expansion}, we obtain \eqref{eq:uK}.
\section{Reweighting the Generated Paths}
\label{sec:reweight}

The controlled dynamics \eqref{eq:general-controlled-sde} generate transition paths more efficiently than the reference dynamics, but at the price of sampling from a tilted path measure rather than the reference one. To fix this, first consider starting from some initial condition $x\in D$. Recall that $P^{x}$ is the path measure of \eqref{eq:ref-sde} started at $x$, and $\tilde P^{*,x}$ is that of the optimally controlled process \eqref{eq:general-controlled-sde} started at $x$. 
By Girsanov's theorem, $P^{x}$ and $\tilde P^{*,x}$ are equivalent masures on the common path space $\Omega$. Let $\chi \in \Omega$ be a generic element of the path space.
By Lemma~\ref{lemma:girsanov-density}, we have
\begin{align}
  \frac{d\tilde P^{*,x}}{dP^{x}}
  = \frac{e^{-H(\chi)}}{\rho(x)},
  \label{eq:reweight-fixed-density}
\end{align}
where $\rho(x)= \E_{x} [e^{-H}]$, and $H(\chi) := \int_0^{\tilde\tau_D} f(\tilde x_t)\,dt + \Phi(\tilde x_{\tilde\tau_D})$. Hence, the correction from the controlled process \eqref{eq:general-controlled-sde} to the original process \eqref{eq:ref-sde} at a fixed start is simply $\rho(x) e^{H(\chi)}$.

When the initial conditions are drawn from a probability measure $\nu$ on the source set $A_i$, e.g., the Boltzmann measure restricted to $A_i$, consider the mixture
\begin{align}
  M(\cdot) = \int_{A_i} P^{x_0}(\cdot)\,\nu(dx_0),
  \quad
   \tilde M(\cdot) = \int_{A_i}\tilde P^{*,x_0}(\cdot)\,\nu(dx_0).
  \label{eq:reweight-mixtures}
\end{align}
Both $M$ and $\tilde M$ share the same initial distribution $\nu$. Since every path starts from its own initial condition $x_0=\chi(0)$, we have
\begin{align}
  \frac{d\tilde M}{dM}
  = \frac{d\tilde P^{*,\chi(0)}}{dP^{\chi(0)}}
  = \frac{e^{-H(\chi)}}{\rho\big(\chi(0)\big)}.
  \label{eq:reweight-mixture-density}
\end{align}
Hence, to map $\tilde M$ back to $M$, each path is weighted by $dM/d\tilde M$ by Lemma~\ref{lemma:girsanov-density}. From \eqref{eq:reweight-mixture-density} and replacing $\rho$ by $\hat{\rho}$ in~\eqref{eq:rho-kernel-expansion}, we can approximate it as
\begin{align}
  \text{weight}(\chi) = \frac{dM}{d\tilde M}
          = \hat{\rho}\big(\tilde x_0\big)\,e^{H(\chi)} ,
  \qquad \tilde x_0 := \chi(0).
  \label{eq:reweight-weight}
\end{align}
Here, $\hat{\rho}(\tilde x_0)$ is given by \eqref{eq:rho-kernel-expansion}.

We further note that as $\beta\to\infty$, the probability that a sampled path spends more than an arbitrary fixed time $\varepsilon>0$ in nonreactive regions tends to zero. To see this, recall from Example~\ref{ex:tps} that $H=\beta H_0$. From~\eqref{eq:reweight-fixed-density}, we have 
\begin{align}
\begin{aligned}
\tilde P^{*,x}(H_0>\varepsilon)
= \frac{\E_x \left[e^{-\beta H_0} \bm{1} \{H_0>\varepsilon\}\right]}
        {\E_x \left[e^{-\beta H_0}\right]} 
\leq&  \frac{e^{-\beta \varepsilon}}
        {\E_x \left[e^{-\beta H_0}\right]} \\
\leq&  \frac{e^{-\beta\varepsilon}}
        {\E_x \left[e^{-\beta H_0} \bm{1}\{H_0\leq\varepsilon/2\}\right]} \\
\leq & \frac{e^{-\beta\varepsilon/2}}{P^x(H_0\le\varepsilon/2)}, 
\end{aligned}
\end{align}
for $\varepsilon>0$ with $P^x(H_0\leq\varepsilon/2)>0$.
Here, the first inequality follows from $e^{-\beta H_0}\leq e^{-\beta\varepsilon}$ on $\{H_0>\varepsilon\}$. Taking $\beta\to\infty$, we have $\tilde P^{*,x}(H_0>\varepsilon)\to 0$.

\section{Experiment Details and Additional Examples}
\label{sec:experiment}
\subsection{The Double-Well System}
\label{appendix:dw}
For the two-dimensional double-well system, the drift and diffusion coefficients in \eqref{eq:ref-sde} are specialized as $b(x) = -\nabla U(x)$, with
\begin{align}
U(x,y) = \tfrac16\Bigl(4(1-x^2-y^2)^2 + 2(x^2-2)^2 + \bigl((x+y)^2-1\bigr)^2 + \bigl((x-y)^2-1\bigr)^2 - 2\Bigr),
\end{align}
and $\sigma(x) = \sigma_0 I$ with $\sigma_0=\sqrt{2k_B T_{\mathrm{emp}}}\approx 0.455$ at $T_{\mathrm{emp}} = 1200\,\mathrm{K}$. Its potential has two minima at $(\pm\sqrt{5}/2,\,0)$ and two channels around $(0,\pm 1)$. Following Algorithm~\ref{sec:algorithm-et}, we compute the Koopman eigenfunctions and extract the metastable cores $A_{\mathrm{left}}$ and $A_{\mathrm{right}}$ as described in Section~\ref{sec:committor}. Specifically, $A_{\mathrm{left}}$ and $A_{\mathrm{right}}$ are obtained by retaining those points whose spectral embedding lies within radius $\varepsilon = 0.5$ of its cluster centroid.
We sample $1000$ uncontrolled and $1000$ controlled paths, all initialized at the minimum of $A_{\mathrm{left}}$. A path is a hit at its first entrance to $A_{\mathrm{right}}$, and is truncated at the numerical censor $T_{\mathrm{max}}$ otherwise. 
We use a step size of $\Delta t=0.01$, and $T_{\mathrm{max}}=10$, equivalently, a budget of $T_{\mathrm{max}}/\Delta t = 1000$ integration steps, and $\beta=1$.
We adopt two metrics to evaluate performance: the target hit percentage (THP), which measures the proportion of trajectories that successfully reach the target before $T_{\mathrm{max}}$; and the energy of the transition state (ETS), which measures the maximum potential energy along each transition path.
As shown in Figure~\ref{fig:2ch-paths} and Table~\ref{tab:2channel_benchmark}, the controller raises the THP from $0\%$ to $99.8\%$, and the controlled ensemble retains both channels rather than collapsing onto one.

\begin{table}[ht]
\centering
\caption{Target hit percentage (THP) and energy of the transition state (ETS) for the 2-channel double well. Uncontrolled and controlled rows are measured in our setup; TPS-DPS values are as reported by \citet{seong2024tpsdps}.}
\label{tab:2channel_benchmark}
\begin{tabular}{lcc}
\toprule
Method & THP ($\uparrow$) & ETS ($\downarrow$) \\
 & \% & kJmol$^{-1}$ \\
\midrule
Uncontrolled (1200K)  & $0.0$  & -- \\
Uncontrolled (2400K)  & $2.7$  & $1.73 \pm 0.42$ \\
Uncontrolled (3600K)  & $18.8$ & $2.07 \pm 0.46$ \\
Uncontrolled (4800K)  & $42.1$ & $2.49 \pm 0.66$ \\
TPS-DPS (1200K)           & $99.9$ & $1.38 \pm 0.16$ \\
Proposed method (1200K) & $\mathbf{99.8}$ & $1.40 \pm 0.11$ \\
\bottomrule
\end{tabular}
\end{table}

\subsection{Muller-Brown Potential}
The M\"uller-Brown potential \citep{Müller1979} is described by 
\begin{align}
    U(x, y) = \frac{1}{10}\sum_{k = 1}^{4}A_{k}\exp\left(a_{k}(x - x_{k})^{2} + b_{k}(x - x{k})(y - y_{k}) + c_{k}(y - y_{k})^{2} \right),
\end{align}
with parameters specified as 
\begin{align}
    \begin{pmatrix}
        A_{1} & a_{1} & b_{1} & c_{1} & x_{1} & y_{1}\\
        A_{2} & a_{2} & b_{2} & c_{2} & x_{2} & y_{2}\\
        A_{3} & a_{3} & b_{3} & c_{3} & x_{3} & y_{3}\\
        A_{4} & a_{4} & b_{4} & c_{4} & x_{4} & y_{4}\\
    \end{pmatrix} = \begin{pmatrix}
        -200 & -1 & 0 & -10 & 1 & 0\\
        -100 & -1 & 0 & -10 & 0 & 0.5\\
        -170 & -6.5 & 11 & -6.5 & -0.5 & 1.5\\
        15 & 0.7 & 0.6 & 0.7 & -1 & -1\\
    \end{pmatrix}
\end{align}

We consider a particle with $m = 12 \mathrm{a.m.u}$ and $\gamma = 1 \mathrm{ps}^{-1}$. As shown in Figure~\ref{fig:muller_brown_paths}, the dynamics has one deep well at $(-0.558, 1.442)$, a medium-deep well at $(0.623, 0.028)$, and a shallow well sitting between the two at $(-0.050, 0.467)$. The goal is to transition from the deep well to the medium well. The source and target cores are defined with radius $\varepsilon = 0.2$. We choose a stepsize of $\Delta t = 0.002 \mathrm{ps}$ and $T_\mathrm{max} = 10\mathrm{ps}$.  
We use a Gaussian kernel with bandwidth $0.35$. 
At temperature $T_\mathrm{emp} = 200K$, the controller raises the THP from $1.8\%$ for the uncontrolled process to $99.0\%$.

\begin{table}[t]
\centering
\caption{M\"uller-Brown benchmark at $T = 150\,\mathrm{K}$ and $200\,\mathrm{K}$.}
\label{tab:mb_benchmark}
\begin{tabular}{lcc}
\toprule
Method & THP ($\uparrow$) & ETS ($\downarrow$)\\
 &  \%  & kJmol$^{-1}$\\
 \midrule
 Uncontrolled (150K)  & $0.1$  & $12.4$\\
 Uncontrolled (200K)      & $1.8$ &  $14.0\pm 1.3$\\
 TPS-DPS (150K) & $100$  &  $15.9\pm 3.5$\\
 TPS-DPS (200K) & $100$  &  $15.7 \pm 1.9$\\
 Proposed method (150K, $\beta = 1$) & $95.1$  &  $14.1 \pm 2.1$\\
 Proposed method (150K, $\beta = 2$) & $95.9$ & $14.3\pm 2.1$\\
 Proposed method (200K, $\beta = 1$) & $98.4$ & $14.5 \pm 2.1$\\
 Proposed method (200K, $\beta = 2$) & $\mathbf{99.0}$ & $14.5 \pm 2.1$\\
\bottomrule
\end{tabular}
\end{table}

\subsection{Alanine Dipeptide}
\label{appendix:aldp}

We simulate alanine dipeptide (ACE--ALA--NME) in vacuum with AMBER99SB-ILDN, hydrogen-bond constraints, and the VVVR integrator at $300\,\mathrm{K}$, friction $10^{-3}\,\mathrm{fs}^{-1}$, and timestep $1\,\mathrm{fs}$. Trajectories start from the energy-minimized $C5$ structure. Following~\citep{seong2024tpsdps}, the target is defined as $\Bcal =\{\mathbf{R}\mid||\xi(\mathbf{R})-\xi(\mathbf{R}_\mathcal B)||<0.75\}$ (in radians), where $\xi=(\phi,\psi)$ are the backbone dihedral angles and $\mathbf{R}_{\mathcal B}$ is the energy-minimized $C7_{\mathrm{ax}}$ structure. A trajectory stops when it enters $\mathcal{B}$ or reaches $T_{\mathrm{max}}=1\mathrm{ps}$. 

The controller $\hat u=\nabla \log \hat\rho$ is evaluated on the Kabsch-aligned Cartesian configuration and applied as an external force $F_{\mathrm{bias}} = s\,\hat u$ in OpenMM, with $s=30\,\mathrm{kJ\,mol^{-1}nm^{-1}}$ and $\|\hat u\|$ capped at $75\,\mathrm{kJ\,mol^{-1}nm^{-1}}$ per atom. We use $\beta=1$, Gaussian bandwidth $1$, and $\varepsilon=0.01$, with $80$ kernel centers and $400$ interior and $48$ boundary collocation points. OpenMM integrates the underdamped Langevin dynamics
$$dx_{i} = \frac{p_{i}}{m_{i}}dt,\quad dp_{i} = F_{i}(x)dt - \gamma_{i}p_{i}dt + \sqrt{2m_{i}\gamma_{i}k_{B}T}dW_{i},$$
where $i$ indexes the atom.
In this experiment, we solve the boundary value problem~\eqref{eq:elliptic-bvp} with the isotropic generator $L = F\cdot\nabla + \tfrac{\sigma_0^2}{2}\Delta$ where $F$ is the force and $\sigma_0=1$ as an approximation of the simulated dynamics that ignores the per-atom mobility $1/(m_i\gamma)$.

The proposed pipeline as summarized in Algorithm \ref{sec:algorithm-et} runs once: Koopman eigenfunctions are constructed via a generalized eigenproblem in terms of Gram matrices (see, e.g., \citep[Appendix M]{hou2023sparse} for details), followed by solving for $\hat{q}$ and $\hat{u}$. In contrast, the baselines in~\citep{seong2024tpsdps} repeatedly simulate controlled trajectories and update the neural network, so they pay the sampling cost per rollout. Table~\ref{tab:cost} compares the cost of obtaining the bias in the drift term. All timings are measured on a single NVIDIA A100 GPU.

\begin{figure}[ht]
  \centering
  \begin{subfigure}[t]{0.45\linewidth}
    \centering
    \includegraphics[height=4cm]{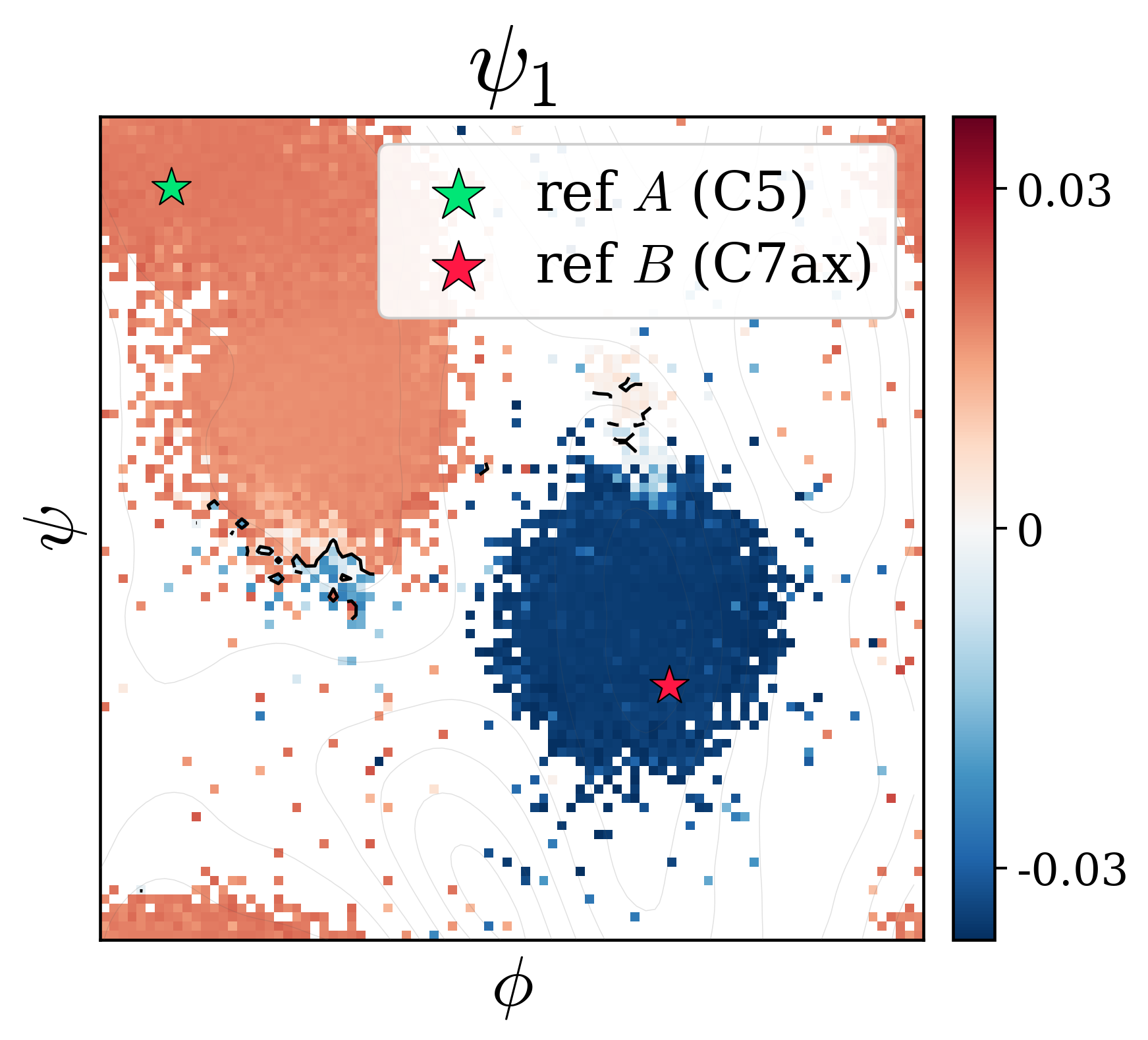}
  \end{subfigure}
  \begin{subfigure}[t]{0.45\linewidth}
    \centering
    \includegraphics[height=4cm]{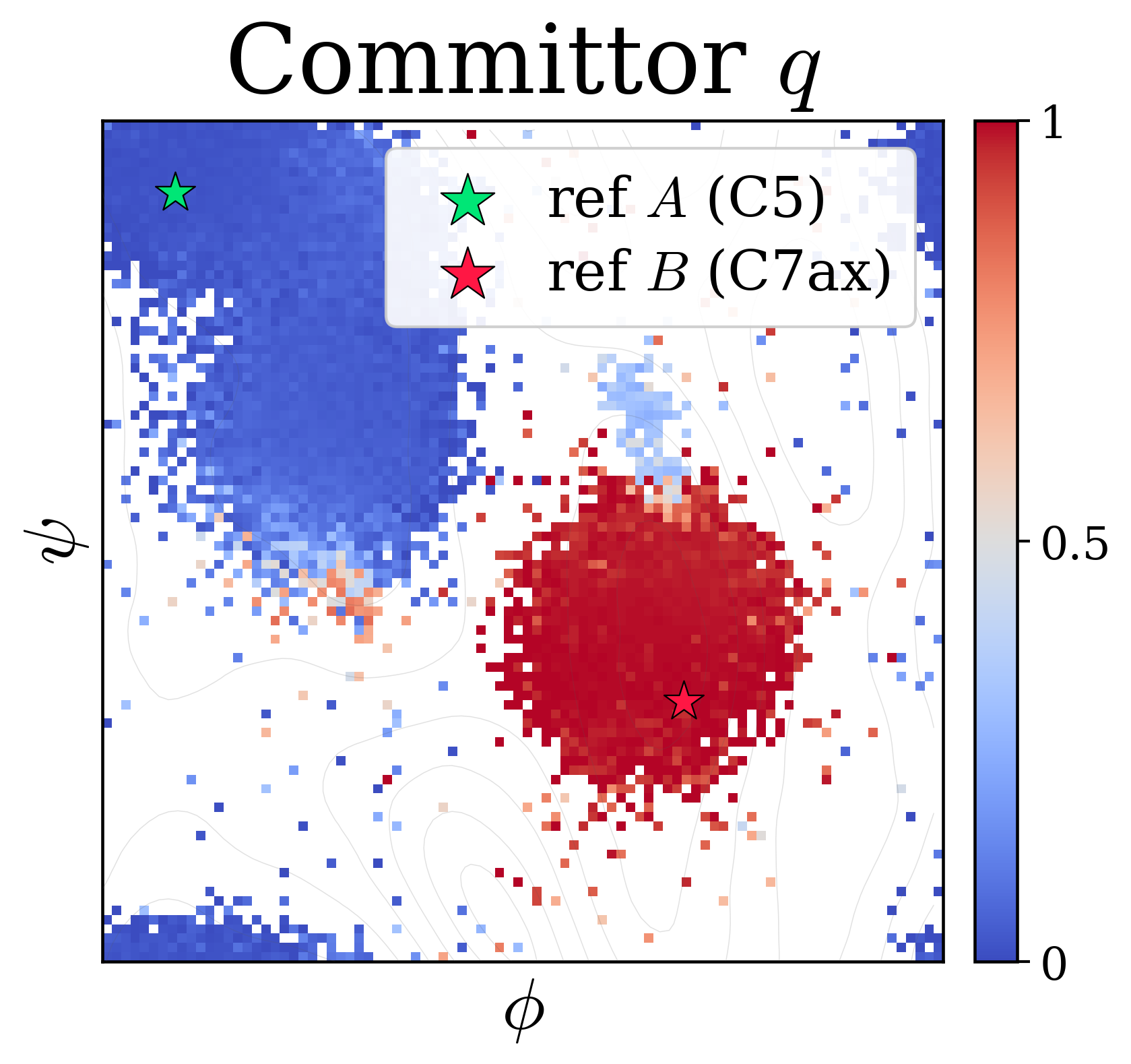}
  \end{subfigure}
  \caption{Alanine dipeptide on the Ramachandran plain of backbone dihedral angles $(\phi,\psi)$. (Left) The first nontrivial Koopman eigenfunction $\psi_1$ which separates the $C5$ basin from $C7_{\mathrm{ax}}$ basin. (Right) The committor $q$ from $C5$  to $C7_{\mathrm{ax}}$.}
  \label{fig:aldp_eigenfunction}
\end{figure}

\begin{figure}[t]
    \centering
    \begin{subfigure}[b]{0.3\textwidth}
        \centering
        \includegraphics[height=4cm]{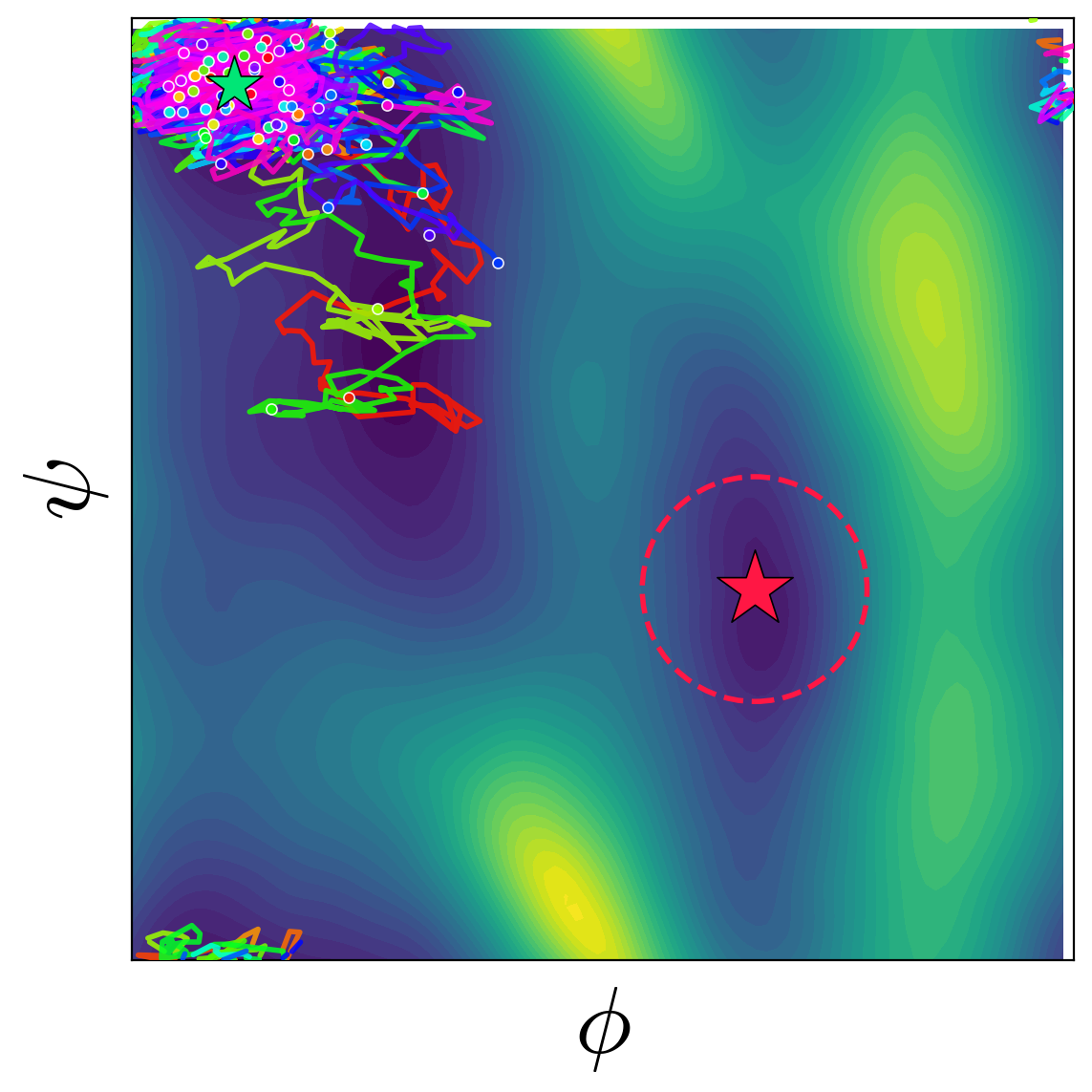}
        \caption{Uncontrolled}
        \label{fig:sub_a}
    \end{subfigure}
    \hfill
    \begin{subfigure}[b]{0.3\textwidth}
        \centering
        \includegraphics[height=4cm]{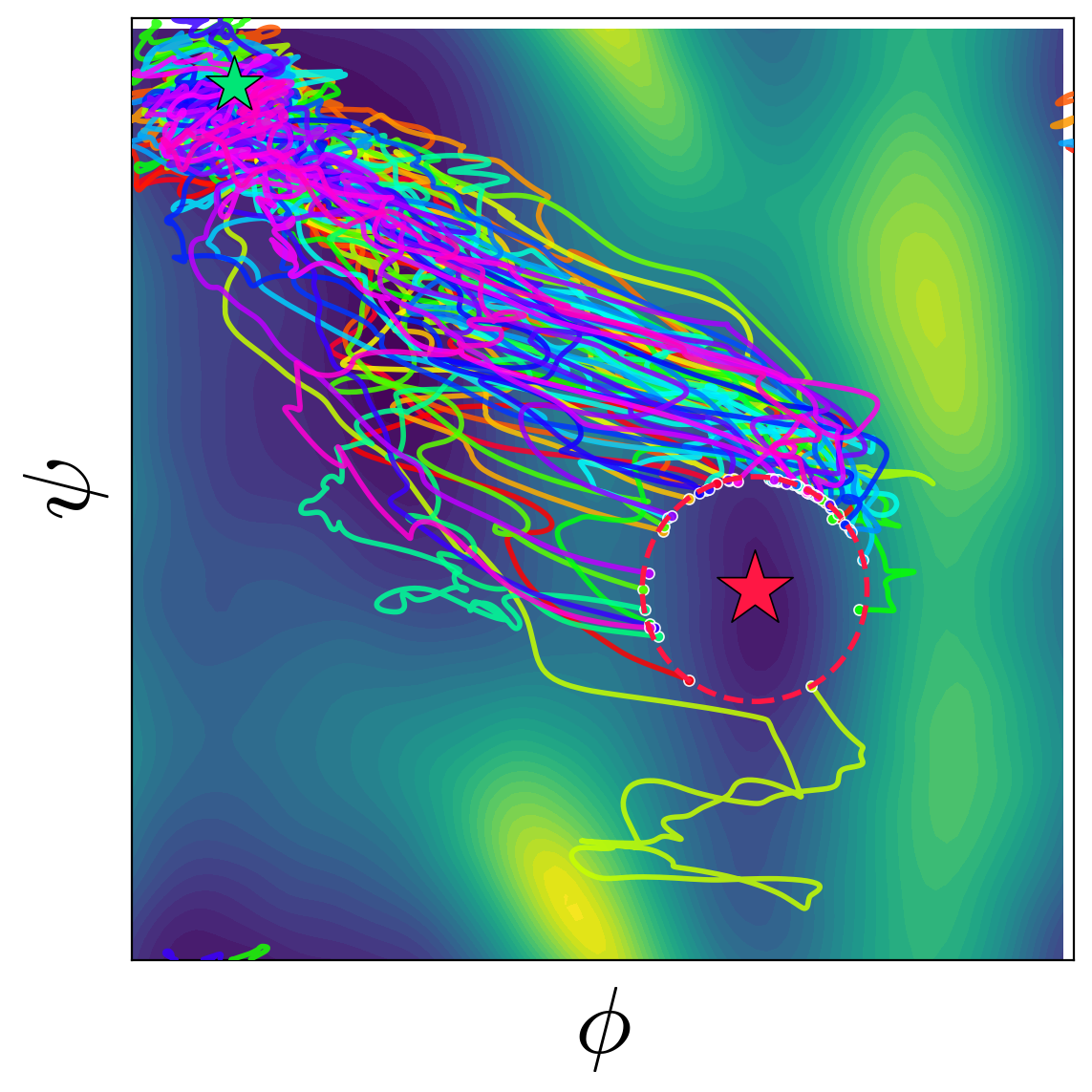}
        \caption{TPS-DPS}
        \label{fig:sub_a}
    \end{subfigure}
    \hfill
    \begin{subfigure}[b]{0.35\textwidth}
        \centering
        \includegraphics[height=4cm]{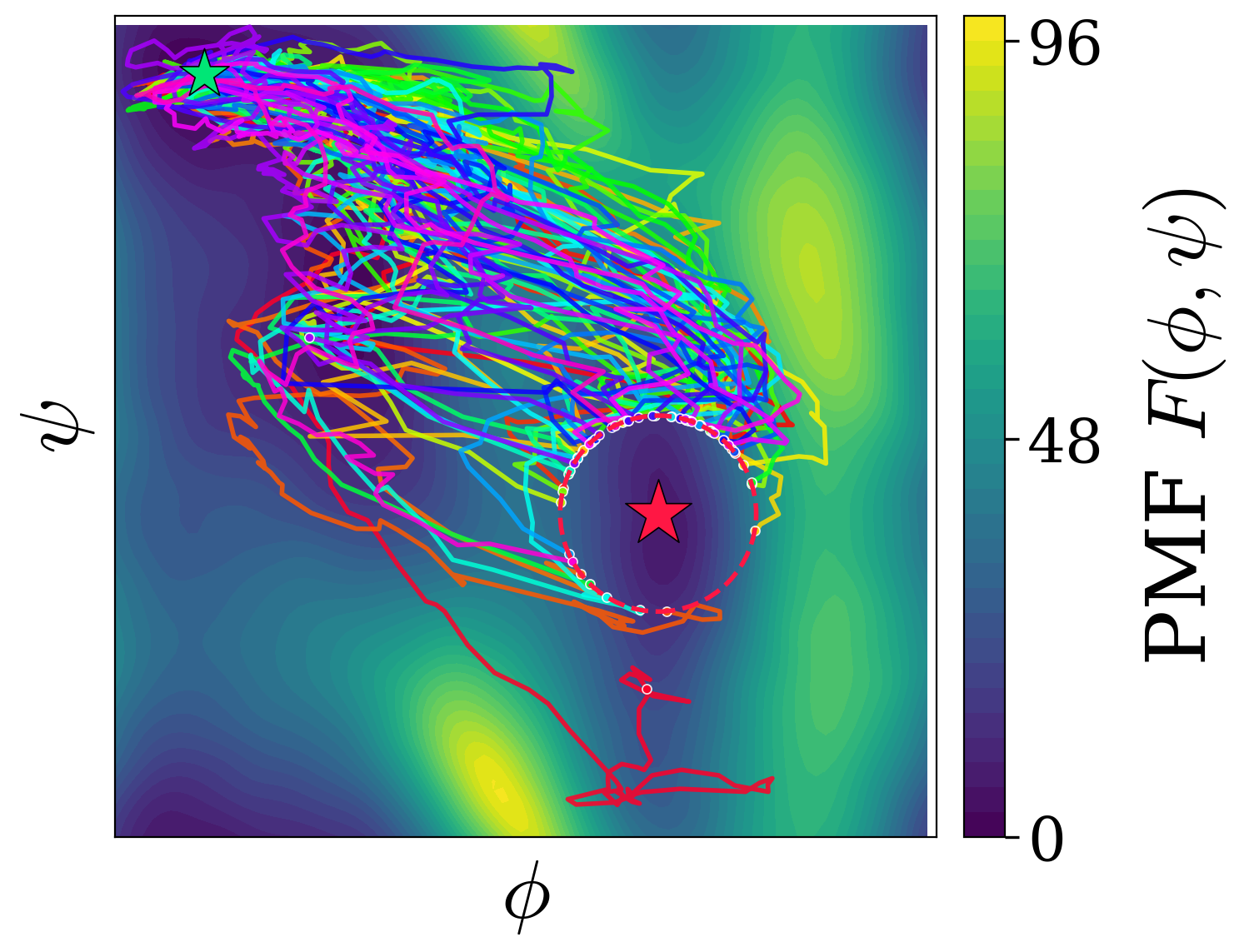}
        \caption{Proposed method}
        \label{fig:sub_b}
    \end{subfigure}
    \caption{Sampled transition paths for alanine dipeptide in vacuum, same as Figure~\ref{fig:alanine_sample_paths} in Section~\ref{sec5}. Paths start in the $C5$ basin (upper left) and end in the $C7_{\mathrm{ax}}$ basin (red star). (a) Uncontrolled paths. (b) TPS-DPS~\citep{seong2024tpsdps}: THP $= 100\%$, obtained by training a neural-network bias with simulation in the loop ($1000$ rollouts, ${\sim}10^6$ gradient updates). (c) Proposed method: THP $= 93\%$, with a closed-form controller.}
    \label{fig:sample_paths.alpd.2}
\end{figure}

\begin{table}[ht]
\centering
\caption{Cost of obtaining the controller on alanine dipeptide on a single NVIDIA A100 GPU, excluding the subsequent path generation, which is common to all methods. }
\label{tab:cost}
\begin{tabular}{lrcrcrcr}
\toprule
\multicolumn{8}{l}{\textit{Proposed method}} \\
\cmidrule(lr){1-8}
Stage & \multicolumn{5}{c}{Time (s)} & & Total (s) \\
Sampling       & \multicolumn{5}{c}{\textcolor{black}{Direct evaluation of \eqref{eq:ref-generator}}} & & \multirow{4}{*}{180} \\
Koopman eigenpairs   & \multicolumn{5}{c}{73.8} & & \\
Committor $\hat q$   & \multicolumn{5}{c}{56.2} & & \\
Controller $\hat u$  & \multicolumn{5}{c}{51.0} & & \\
\midrule
\multicolumn{8}{l}{\textit{Training-based baselines}} \\
\cmidrule(lr){1-8}
Method & Rollouts & $\times$ & (\,Sampling (s) & $+$ & Training (s)\,) & $=$ & Total (h) \\
TPS-DPS (F) & 1000  & $\times$ & (\,60.6 & $+$ & 14\,) & $=$ & 20.7 \\
PIPS (F)    & 15000 & $\times$ & (\,60.6 & $+$ & 5.4\,) & $=$ & 275 \\
\bottomrule
\end{tabular}
\end{table}

\subsection{Additional Experiment: The Four-Well Dynamics}
\label{sec:experiments}
Consider the two-dimensional four-well potential described by 
\begin{align}
    U(x, y) = 5((x^{2} - 1)^{2} + (y^{2} - 1)^{2}),
\end{align}
which has four symmetric local minima at $(\pm 1, \pm 1)$. The source core is centered at pair $(-1, -1)$ and the target core is centered at $(1, 1)$, with radius $\varepsilon =0.3$. We use mass $m = 12$ a.m.u, friction $\gamma = 1\mathrm{ps}^{-1}$, stepsize $\Delta t = 0.002\mathrm{ps}$, and $T_{\mathrm{max}} = 10\mathrm{ps}$. Results for various settings are reported in~\ref{tab:4well_benchmark}, and the sampled paths at $200\,\mathrm{K}$ are shown in Figure~\ref{fig:four_well_paths_T_200}.

\begin{figure}[ht]
    \centering
    \begin{subfigure}[b]{0.3\textwidth}
    \centering
        \includegraphics[height = 4cm]{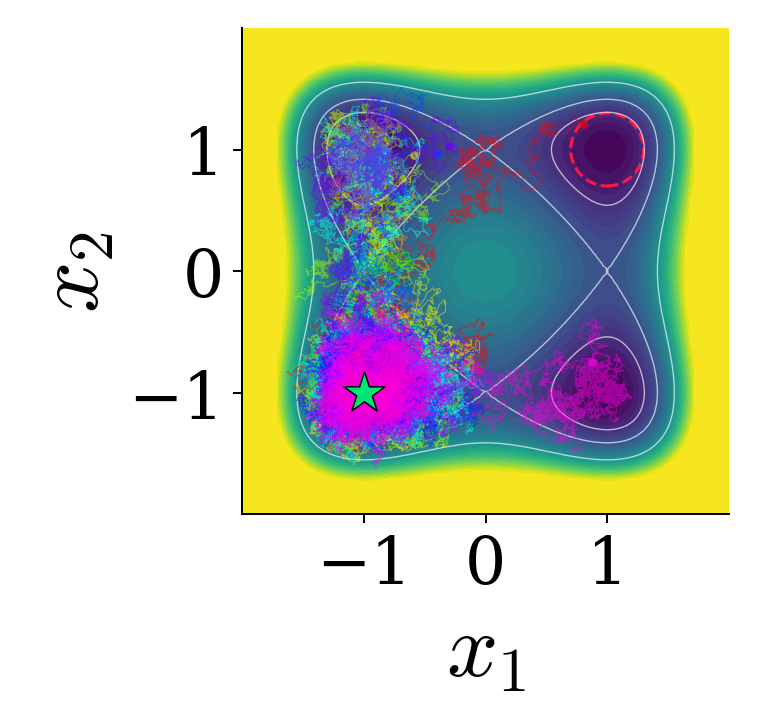}
        \caption{Uncontrolled}
        \label{fig:sub_a_4w_paths}
    \end{subfigure}
    \hfill
    \begin{subfigure}[b]{0.3\textwidth}
    \centering
        \includegraphics[height = 4cm]{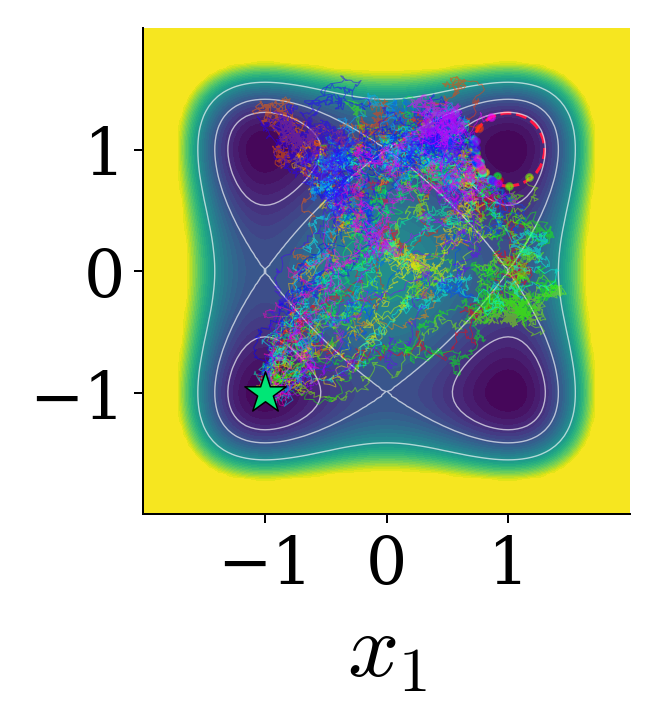}
        \caption{TPS-DPS}
        \label{fig:sub_b_4w_paths}
    \end{subfigure}
    \hfill
    \begin{subfigure}[b]{0.3\textwidth}
    \centering
        \includegraphics[height = 4cm]{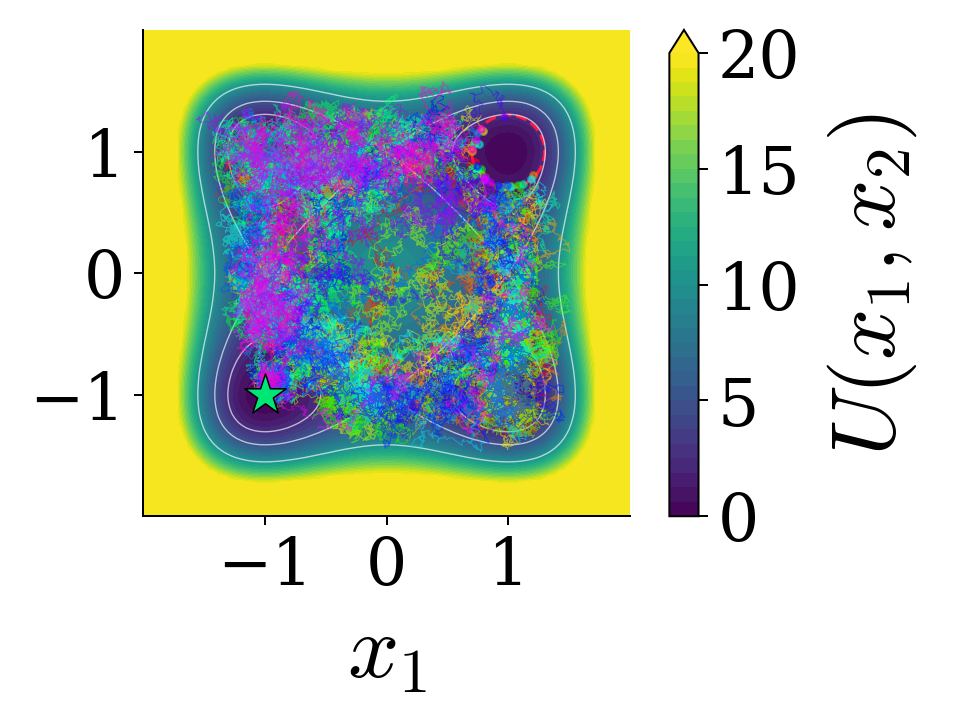}
        \caption{Proposed method}
        \label{fig:sub_c_4w_paths}
    \end{subfigure}
    \caption{Four-well dynamics at $T = 200\,\mathrm{K}$. We aim to sample transitions from the source core at $(-1,-1)$ (green star) to the target core at $(1,1)$ (red star). (a) Uncontrolled paths: THP($T_\mathrm{max}$) = $0.9\%$. (b) TPS-DPS~\citep{seong2024tpsdps}: THP($T_\mathrm{max}$) = $81.6\%$. (c) Proposed method with $\beta = 2$: THP($T_{\mathrm{max}}$) = $94.7\%$.}
    \label{fig:four_well_paths_T_200}
\end{figure}

\begin{table}[t]
\centering
\caption{Four-well benchmark at $T = 200\,\mathrm{K}$ and $300\,\mathrm{K}$.}
\label{tab:4well_benchmark}
\begin{tabular}{lcc}
\toprule
Method & THP ($\uparrow$) & ETS ($\downarrow$)\\
 &  \%  & kJmol$^{-1}$\\
 \midrule
 Uncontrolled (200K)  & $0.9$  & $10.6\pm 3.2$\\
 Uncontrolled (300K)      & $8.8$ &  $12.8\pm 3.6$\\
 TPS-DPS (200K) & $81.6$  &  $8.6\pm .9$\\
 TPS-DPS (300K) & $87.1$  &  $10.2 pm .6$\\
 Proposed method (200K, $\beta = 1$) & $73.2$  &  $9.5 \pm 1.8$\\
 Proposed method (200K, $\beta = 2$) & $94.7$ & $9.0\pm 1.8$\\
 Proposed method (300K, $\beta = 1$) & $91.8$ & $11.2 \pm 3.1$\\
 Proposed method (300K, $\beta = 2$) & $\mathbf{100.0}$ & $10.4 \pm 2.5$\\
\bottomrule
\end{tabular}
\end{table}

\end{document}